\documentclass[%
 reprint,
 superscriptaddress,
 amsmath,amssymb,
 aps,
 pra,
]{revtex4-2}

\usepackage[utf8]{inputenc} 
\usepackage{amsmath,amsfonts,amssymb,mathrsfs,amsthm,cases} 
\usepackage{bm, bbm, dsfont} 
\usepackage{physics} 
\usepackage{graphicx} 
\usepackage{extarrows}
\usepackage{diagbox,threeparttable,dashrule,booktabs,dcolumn} 
\usepackage{xcolor} 
\usepackage{url} 
\usepackage{hyperref} 
\hypersetup{colorlinks,linkcolor={blue},citecolor={blue},urlcolor={red}}
\usepackage{cleveref} 

\crefname{equation}{Eq.}{Eqs.}
\crefname{section}{Sec.}{Secs.}
\crefname{definition}{Definition}{Definitions}
\crefname{proposition}{Proposition}{Propositions}
\crefname{lemma}{Lemma}{Lemmas}
\crefname{theorem}{Theorem}{Theorems}
\crefname{corollary}{Corollary}{Corollaries}
\crefname{conjecture}{Conjecture}{Conjectures}
\crefname{claim}{}{Claims}
\crefname{example}{Example}{Examples}

\newtheorem{definition}{Definition}
\newtheorem{proposition}[definition]{Proposition}
\newtheorem{lemma}[definition]{Lemma}

\newtheorem{theorem}[definition]{Theorem}
\newtheorem{corollary}[definition]{Corollary}

\newtheorem*{remark}{Remark}
\newtheorem{example}[definition]{Example}

\newtheorem{manualtheoreminner}{Corollary}
\newenvironment{manualtheorem}[1]{%
  \renewcommand\themanualtheoreminner{#1}%
  \manualtheoreminner
}{\endmanualtheoreminner}

\newcommand{\proj}[1]{|{#1}\rangle \langle {#1}|}

\newcommand{\nc}{\newcommand}

\def\bpf{\begin{proof}}
\def\epf{\end{proof}}
\def\bea{\begin{eqnarray}}
\def\eea{\end{eqnarray}}
\def\beq{\begin{equation}}
\def\eeq{\end{equation}}
\def\bal{\begin{aligned}}
\def\eal{\end{aligned}}
\def\bma{\begin{pmatrix}}
\def\ema{\end{pmatrix}}

\def\dg{\dagger}

\def\ox{\otimes}

\def\lin{\mathop{\rm span}}
\def\I{\mathds{I}}
\def\diag{\rm diag}

\def\a{\alpha}

\def\t{\theta}

\nc{\bbA}{{\mathbb A}}  \nc{\bbB}{{\mathbb B}}  \nc{\bbC}{{\mathbb C}}
\nc{\bbD}{{\mathbb D}}  \nc{\bbE}{{\mathbb E}}  \nc{\bbF}{{\mathbb F}}
\nc{\bbG}{{\mathbb G}}  \nc{\bbH}{{\mathbb H}}  \nc{\bbI}{{\mathbb I}}
\nc{\bbJ}{{\mathbb J}}  \nc{\bbK}{{\mathbb K}}  \nc{\bbL}{{\mathbb L}}
\nc{\bbM}{{\mathbb M}}  \nc{\bbN}{{\mathbb N}}  \nc{\bbO}{{\mathbb O}}
\nc{\bbP}{{\mathbb P}}  \nc{\bbQ}{{\mathbb Q}}  \nc{\bbR}{{\mathbb R}}
\nc{\bbS}{{\mathbb S}}  \nc{\bbT}{{\mathbb T}}  \nc{\bbU}{{\mathbb U}}
\nc{\bbV}{{\mathbb V}}  \nc{\bbW}{{\mathbb W}}  \nc{\bbX}{{\mathbb X}}
\nc{\bbY}{{\mathbb Y}}  \nc{\bbZ}{{\mathbb Z}}  

\nc{\bA}{{\bf A}}  \nc{\bB}{{\bf B}}  \nc{\bC}{{\bf C}}
\nc{\bD}{{\bf D}}  \nc{\bE}{{\bf E}}  \nc{\bF}{{\bf F}}
\nc{\bG}{{\bf G}}  \nc{\bH}{{\bf H}}  \nc{\bI}{{\bf I}}
\nc{\bJ}{{\bf J}}  \nc{\bK}{{\bf K}}  \nc{\bL}{{\bf L}}
\nc{\bM}{{\bf M}}  \nc{\bN}{{\bf N}}  \nc{\bO}{{\bf O}}
\nc{\bP}{{\bf P}}  \nc{\bQ}{{\bf Q}}  \nc{\bR}{{\bf R}}
\nc{\bS}{{\bf S}}  \nc{\bT}{{\bf T}}  \nc{\bU}{{\bf U}}
\nc{\bV}{{\bf V}}  \nc{\bW}{{\bf W}}  \nc{\bX}{{\bf X}}
\nc{\bY}{{\bf Y}}  \nc{\bZ}{{\bf Z}}  

\nc{\cA}{{\cal A}}  \nc{\cB}{{\cal B}}  \nc{\cC}{{\cal C}}
\nc{\cD}{{\cal D}}  \nc{\cE}{{\cal E}}  \nc{\cF}{{\cal F}}
\nc{\cG}{{\cal G}}  \nc{\cH}{{\cal H}}  \nc{\cI}{{\cal I}}
\nc{\cJ}{{\cal J}}  \nc{\cK}{{\cal K}}  \nc{\cL}{{\cal L}}
\nc{\cM}{{\cal M}}  \nc{\cN}{{\cal N}}  \nc{\cO}{{\cal O}}
\nc{\cP}{{\cal P}}  \nc{\cQ}{{\cal Q}}  \nc{\cR}{{\cal R}}
\nc{\cS}{{\cal S}}  \nc{\cT}{{\cal T}}  \nc{\cU}{{\cal U}}
\nc{\cV}{{\cal V}}  \nc{\cW}{{\cal W}}  \nc{\cX}{{\cal X}}
\nc{\cY}{{\cal Y}}  \nc{\cZ}{{\cal Z}}  

\def\ox{\otimes}

\def\dg{\dagger}

\def\lin{\mathop{\rm Span}}

\begin{document}


\title{Decidabilities of local unitary equivalence for entanglement witnesses and states} 

\author{Yi Shen}
\email[]{yishen@jiangnan.edu.cn}
\affiliation{School of Science, Jiangnan University, Wuxi Jiangsu 214122, China}

\author{Lin Chen}
\email[Corresponding author: ]{linchen@buaa.edu.cn}
\affiliation{School of Mathematical Sciences, Beihang University, Beijing 100191, China}

\date{\today} 

\begin{abstract}
The problem of determining whether two states are equivalent by local unitary (LU) operations is important for quantum information processing. In this paper we propose an alternative perspective to study this problem by comparing the decidabilities of LU equivalence (also known as LU decidabilities for short) between entanglement witnesses and states. We introduce a relation between sets of Hermitian operators in terms of the LU decidability. Then we compare the LU decidability for the set of entanglement witness to those decidabilities for several sets of states, and establish a hierarchy on LU decidabilities for these sets. Moreover, we realize that the simultaneous LU (SLU) equivalence between tuples of mutually orthogonal projectors is crucial to LU equivalent operators. We reveal by examples that for two tuples of projectors, the partial SLU equivalence cannot ensure the overall SLU equivalence. Generally, we present a necessary and sufficient condition such that two tuples of states are SLU equivalent.
\end{abstract}


\maketitle


\section{Introduction}
\label{sec:intro}

At the heart of quantum mechanics lies the concept of entanglement, a phenomenon that distinguishes quantum systems from classical ones, and serves as a key resource for various information processing protocols. An efficient way to characterize entanglement is to classify entangled states under several equivalence relations relying on different kinds of local operations such as local unitary (LU) operations \cite{pslu-2010}, local operations and classical communication (LOCC) \cite{3qubitinequiv2000}, and stochastic LOCC (SLOCC) \cite{MCSLOCC2011}. The fundamental one is LU equivalence, which refers to the equivalence relation between two multipartite states if they can be transformed into each other by applying LU operations. This equivalence plays an essential role in understanding the structures and properties of entanglement, as it defines an equivalence class of states that share identical entanglement characteristics under local operations. In addition, since the process of LU transformation is reversible, the LU classification of states is thus crucial for determining the resource interconvertibility among different entangled states \cite{rvqrt2015}, and has significant implications for quantum computing \cite{luqc-2022}, communication \cite{qcLOCC2021} and metrology \cite{metrology2016}.

For bipartite pure states, the Schmidt decomposition provides an effective technique to classify them under LU equivalence. Although it is lack of a generalization of the Schmidt decomposition for multipartite scenarios, the complete classification for multipartite pure states under LU equivalence has been progressively developed \cite{3qubitlu2000,luequiv2010,mpsinequivlu2012}. 
Despite the breakthroughs on the LU classification of pure states, the problem of characterizing LU orbits of mixed states remains challenging. A possible way to attack this problem is to discover algebraic invariants \cite{Makhlin2002} under LU transformations which serve as necessary conditions for determining whether two states are LU equivalent. Studying along this way has been of great interest in recent years \cite{msinv2013,mqbitinv2014,mqbitiff2015,mqbitinv2015,3qbitinv2017,msinv2024,lininv2024}.
However, many open questions remain, especially regarding the completeness of invariant sets and the practical implementation of classification algorithms for higher-dimensional systems.

Indeed, the classification under LU equivalence is not confined to states. The study has been extended to several other objects in entanglement theory, e.g. the non-local unitary operations \cite{decugate15,ugatesyi22} and the entanglement witnesses (EWs) \cite{ewreview2014,rew2024}. The characterizations of non-local unitary operations and EWs under LU equivalence are also of physical significance \cite{ep2016-2,mep2025,inertia24}. From an operational perspective, to implement a complicated non-local unitary gate or EW, an efficient approach is to first implement a simpler one which is LU equivalent to the target, assisted by a sequence of LU operations \cite{eff-nug-2015,lewo2020}. An entangled state with NPT (non-positive partial transpose) property can be linked to an EW by virtue of the partial transposition. This connection motivates us to jointly investigate the LU equivalence of EWs and that of states. Since EWs and states are both expressed by Hermitian operators which have spectral decompositions, some EWs and entangled states can be mutually converted by adjusting eigenvalues. Thus, we propose an alternative perspective on the LU equivalence of states by connecting with the LU equivalence of EWs. 


In this paper we introduce the concept of LU decidable sets where any two bipartite operators can be decided as LU equivalent or LU inequivalent. By comparing the LU decidabilities between two sets, we further introduce a relation in terms of LU decidabilities. For two sets $\cS_1,\cS_2$ of bipartite operators, we use $\cS_1\preceq\cS_2$ to indicate that the LU decidability of $\cS_1$ is not stronger than that of $\cS_2$, which alternatively means that $\cS_1$ is LU decidable if $\cS_2$ is LU decidable. By virtue of this relation we connect the LU decidability of the set of EWs with those of several sets of states. Based on the comparative results we establish a hierarchy on LU decidabilities for these concerning sets. To visualize this hierarchy, we illustrate it by Figure \ref{fig:1}. In Theorem \ref{cor:drelation} we build the relations among the set of EWs and the sets of states divided by the PPT (positive-partial-transpose) criterion \cite{ppte1997}. In order to more accurately evaluate the LU decidability of the set of EWs, we figure out a set of states each of which satisfies that there is an eigenspace containing no product vector, and show by Theorem \ref{thm:mainrelation} that the LU decidability of such set is equal to that of the set of EWs. Furthermore, we extend our study to a larger set of states each of which satisfies that there exists some eigenspace not spanned by product vectors. We conclude in Proposition \ref{prop:kprodv} that the LU decidability of this larger set is equal to that of the set above-mentioned. By analyzing the spectral decompositions of LU equivalent operators, we realize that the SLU equivalence between tuples of mutually orthogonal projectors plays an essential role. In light of this, we consider the conditions ensuring the SLU equivalence between such two tuples. A necessary condition requires that any parts of such two tuples have to be SLU equivalent, which we call the partial SLU equivalence. By Example \ref{ex:slu-c} we indicate that the partial SLU equivalence cannot ensure the SLU equivalence between the whole tuples of mutually orthogonal projectors. In a more general context, we propose a necessary and sufficient condition in Theorem \ref{le:slu} such that two tuples of states are SLU equivalent.

\begin{figure*}[htbp]
\centering
\includegraphics[width=0.9\textwidth]{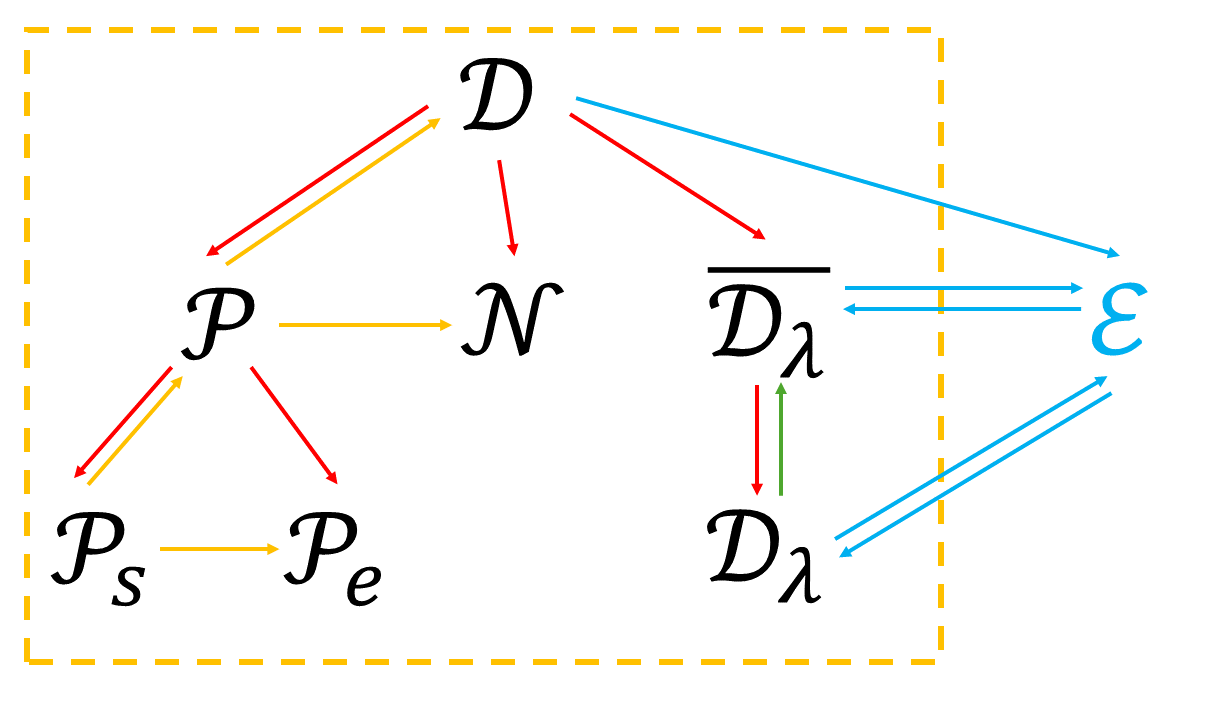}
\caption{This figure depicts the relations in terms of the LU decidability among several sets. Denote by $\cD,\cN,\cP,\cP_e,\cP_s,\cE$, the sets of all states, NPT states, PPT states, PPT entangled states, separable states, and EWs respectively, supported on the same bipartite Hilbert space. We refer to $\cD_\lambda$ and $\overline{\cD_\lambda}$ respectively as the set of bipartite states satisfying that at least one eigenspace contains no product vector, and the set of bipartite states satisfying that at least one eigenspace is not spanned by product vectors. The arrow ``$\rightarrow$'' represents the notation ``$\succeq$'', which means the LU decidability of the set at the end point of the arrow is not stronger than that of the set at the initial point. The arrows in red are concluded from the inclusion relation between the two sets. The arrows in yellow are derived by Theorem \ref{cor:drelation}, the arrows in blue are related to Theorem \ref{thm:mainrelation}, and the arrow in green is derived by Proposition \ref{prop:kprodv}. The arrows in the dashed box reveal the relations among the sets of states, and the arrows in blue establish the connection between $\cE$ and some sets of states.} 
\label{fig:1}
\end{figure*}

The remainder of this paper is organized as follows: In Sec. \ref{sec:pre} we clarify the notations and definitions, and present some known results as necessary preliminaries. In Sec. \ref{sec:EW-orbits} we establish connections of the LU decidabilities among the set of EWs and several sets of states, which also reveals a hierarchy of these sets in terms of the LU decidability. In Sec. \ref{sec:rangelu} we investigate the SLU equivalence between tuples of projectors, which is quite essential to the problem of determining whether two arbitrary states are LU equivalent or not.
Finally, the concluding remarks are given in Sec. \ref{sec:con}.

\section{Preliminaries}
\label{sec:pre}

In this section we perform mathematical preparations to make our discussions clear and concise. 
We first clarify some common notations. The LU equivalence is the core concept in this paper, which implies the mutual transformation of bipartite operators by LU operations. Given two bipartite Hermitian operators $M$ and $N$, we use $M\sim N$ to represent that they are LU equivalent, specifically defined by $M=(U\ox V)N(U\ox V)^\dg$ for some LU operator $U\ox V$. Conversely, we may use $M\nsim N$ to represent that they are LU inequivalent, which means there is no LU operation such that $M, N$ can be transformed to each other. One can analagously formulate the LU equivalence of bipartite subspaces by $X=(U\ox V)Y$ for two subspaces $X,Y$ of $\bbC^m\ox\bbC^n$. We also use $X\sim Y$ and $X\nsim Y$ to represent that such two subspaces are LU equivalent and LU inequivalent respectively. 

For any Hermitian operator $H$, its spectral decomposition reads as $H=\sum_j \lambda_j P_j$, where $\lambda_j$'s are distinct eigenvalues and $P_j$'s are orthogonal projectors. Unless specified otherwise, the eigenvalues in the spectral decomposition are sorted by descending order. We shall use the shorthand notation $\lambda_j$-eigenspace to represent the eigenspace associated with the $j$-th eigenvalue $\lambda_j$.
Unless stated otherwise, the states addressed in this paper are assumed as non-normalized, i.e., the traces of states are just positive.
Denote by $\cD,\cN,\cP,\cP_e,\cP_s,\cE$ the sets of all states, NPT states, PPT states, PPT entangled states, separable states, and EWs respectively, supported on the same bipartite Hilbert space. We shall use $\cR(\rho)$ to represent the range of positive semidefinite $\rho$.

Second, based on LU equivalence, we introduce two kinds of relations between sets of bipartite Hermitian operators as follows. Definition \ref{def:decide} establishes a hierarchy on the decidability of LU equivalence. Definition \ref{def:slu} characterizes whether a set of bipartite Hermitian operators can be correspondingly transformed to another set by a common LU operator.

\begin{definition}
\label{def:decide}
(i) For a given set $\cS$ of bipartite Hermitian operators, if the LU equivalence or LU inequivalence between any two operators in $\cS$ is known, the set $\cS$ is called LU decidable. 

(ii) For two given sets $\cS_1,\cS_2$ of bipartite Hermitian operators, if $\cS_1$ is LU decidable when $\cS_2$ is LU decidable, then the LU decidability of $\cS_1$ is not stronger than that of $\cS_2$, denoted by $\cS_1\preceq\cS_2$.

(iii) We call the LU decidabilities of $\cS_1$ and $\cS_2$ equal, denoted by $\cS_1\xlongequal{d}\cS_2$, if both two relations $\cS_1\preceq\cS_2$ and $\cS_2\preceq\cS_1$ hold simultaneously.
\end{definition}

By definition one can verify that the relation ``$\preceq$'', admits the transitivity. That is, $\cS_1\preceq\cS_3$ holds if both relations $\cS_1\preceq\cS_2$ and $\cS_2\preceq\cS_3$ hold. 
Furthermore, one can check that the relation ``$\xlongequal{d}$'' is an equivalence relation.

Also, the relation of LU decidability is linked to the inclusion relation. We give some remarks on this link as follows. 
For two sets with an inclusion relation, e.g. $\cS_1\subseteq\cS_2$, their relation of LU decidability naturally admits as $\cS_1\preceq\cS_2$. For such two sets, if it is further known that $\cS_1\xlongequal{d}\cS_2$, one can verify that for any set $\cS_3$ with $\cS_1\subseteq\cS_3\subseteq\cS_2$, the equivalence relation $\cS_1\xlongequal{d}\cS_2\xlongequal{d}\cS_3$ holds.
Therefore, it establishes a hierarchy of sets of bipartite Hermitian operators on the LU decidability. 

It is necessary to simultaneously consider the LU equivalence of multiple pairs of bipartite Hermitian operators, in order to determine the relation of LU decidability for two given sets. Therefore, the following introduced simultaneous LU equivalence is essential to study the hierarchy on LU decidability.

\begin{definition}
\label{def:slu}
Suppose that $\cS_1=(M_1,\cdots,M_k)$ and $\cS_2=(N_1,\cdots,N_k)$ are two $k$-tuples of bipartite Hermitian operators supported on $\bbC^m\ox\bbC^n$. We call that $\cS_1$ is simultaneously LU (SLU) equivalent to $\cS_2$, denoted by $\cS_1\sim_s\cS_2$, if there exists a common LU operator $U\ox V$ such that $(U\ox V)M_i(U\ox V)^\dg=N_i$ for each $1\leq i\leq k$.
\end{definition}

The definition above can be directly extended to the context of bipartite subspaces. We here do not repeat the definition of SLU equivalence for bipartite subspaces.

Finally, we present some konwn lemmas which are necessary to derive our main results. Lemma \ref{le:noise-sep} indicates that the entanglement of a state will vanish completely by adding large enough noise. Lemma \ref{le:spectral-dec} proposes a necessary and sufficient condition in terms of SLU equivalence given by Definition \ref{def:slu}, such that two bipartite Hermitian operators are LU equivalent.

\begin{lemma}
\label{le:noise-sep}
Given a multipartite state $\rho$, there is a large enough coefficient $x>0$ such that $\rho+x\I$ is a separable state, where $\I$ is supported on the same space as $\rho$.
\end{lemma}

\begin{lemma}
\label{le:spectral-dec}
Let $H,K$ be two bipartite Hermitian operators. Their spectral decompositions read as $H=\sum_{j=1}^m\lambda_j M_j$ and $K=\sum_{j=1}^n\mu_j N_j$, where $\lambda_j's$ and $\mu_j's$ are non-zero eigenvalues of $H$ and $K$ respectively in descending order. Then $H$ is LU equivalent to $K$, i.e. $H\sim K$, if and only if $m=n$, and for each $j$, $\lambda_j=\mu_j$ and $M_j\sim N_j$ by a common LU operator.
\end{lemma}

It is worthy to note that the SLU equivalence of the two $m$-tuples $(M_j)_{j=1}^m,(N_j)_{j=1}^m$ is indispensable for Lemma \ref{le:spectral-dec}. In other words, for two given pairs of bipartite projectors $(A_1,B_1), (A_2,B_2)$, and two real coefficients $x,y$, the two operators formulated by $\alpha_i:=xA_i+yB_i$ for $i=1,2$ may not be LU equivalent, even if $A_1\sim A_2$ and $B_1\sim B_2$. We explain it by example. Let $A_1=A_2=\ketbra{\psi}$ with $\ket{\psi}=\cos\theta\ket{00}+\sin\theta\ket{11}$ for $\theta\in(0,\frac{\pi}{2})$. Then we have $A_1\sim A_2$. Let $B_1=\ketbra{\Psi^+}$ and $B_2=\ketbra{\Psi^-}$ where $\ket{\Psi^+}=\frac{1}{\sqrt{2}}(\ket{01}+\ket{10})$ and $\ket{\Psi^-}=\frac{1}{\sqrt{2}}(\ket{01}-\ket{10})$. Then we have that $B_1$ and $B_2$ are LU equivalent by the LU operator as $\I_2\otimes\diag(-1,1)$. However, one can verify that the eigenvalues of $\alpha_1$ are not the same as those of $\alpha_2$. Thus, $\a_1$ and $\a_2$ are not LU equivalent. It is also known by Lemma \ref{le:spectral-dec} that $(A_1,B_1)$ is not SLU equivalent to $(A_2,B_2)$.

The following lemma characterizes the unitary matrices which maintain diagonal matrices invariant under the unitary similarity.
\begin{lemma}
\label{le:diag-inv}
Let $\Lambda=\oplus_{j=1}^n\lambda_j\I_{d_j}$ be a diagonal matrix. The unitary matrix $U$ such that $U\Lambda U^\dg=\Lambda$ has to be a direct sum of unitary matrices, namely $U=\oplus_{j=1}^n U_j$ where $U_j$ is supported on $\mathbb C^{d_j}$.
\end{lemma}


\section{relations on LU decidabilities among the set of EWs and sets of states}
\label{sec:EW-orbits}

EWs provide a powerful tool to detect entanglement. Due to the inherent connection between LU operations and entanglement characteristics, it inspires us to jointly consider the LU orbits of EWs and the LU orbits of states. Based on the relation in terms of the LU decidability introduced by Definition \ref{def:decide}, we shall focus on such relations among sets of bipartite states and the set of EWs. Our discussion in this section not only establishes a connection of the two problems on the LU decidabilities respectively for states and EWs, but also reveals a hierarchy of the sets of states and that of EWs. To better understand these relations, one may refer to \ref{fig:1} which has visualized the main results in this section.

In Theorem \ref{cor:drelation} we compare the LU decidabilities among the set of EWs and some common sets of bipartite states, namely $\cN,\cP,\cP_e,\cP_s$. In Theorem \ref{thm:mainrelation} we introduce a subset $\cD_{\lambda}\subset\cD$, and show that the LU decidability of $\cD_{\lambda}$ is equal to that of $\cE$. Such two main results in this section indicates that the LU decidabilities of two sets including different types of Hermitian operators could be related. It sheds new light on the problem of verifying the LU equivalence or inequivalence between two arbitrary states by introducing other types of Hermitian operators, e.g. EWs. Then we extend our study to a larger set of states, denoted by $\overline{\cD_{\lambda}}$, which includes the set $\cD_{\lambda}$, and show by Proposition \ref{prop:kprodv} that the two LU decidabilities of $\cD_{\lambda}$ and $\overline{\cD_{\lambda}}$ are equal.
Finally, by Proposition \ref{le:eigrange} we show two special cases where the LU equivalence between two normalized states can be determined by specific eigenvalues of the partial transposes of such two states.

The PPT criterion is a celebrated tool to detect entanglement, which divides $\cD$ into two non-intersected parts, namely $\cN,\cP$, and tells that each state of $\cN$ is entangled. The set $\cP$ can be further divided to two disjoint subsets $\cP_e$ and $\cP_s$. Here, we first discuss the relations in terms of the LU decidability among these common sets of states, $\cN,\cP,\cP_e,\cP_s$, and the set of EWs $\cE$.

\begin{theorem}
\label{cor:drelation}
The decidabilities of LU equivalence for the set of EWs and some common sets of states admit the following chain of relations:
\beq
\label{eq:chain-1}
\cN\preceq \cE \preceq \cD \xlongequal{d} \cP_s \xlongequal{d} \cP \succeq \cP_e.
\eeq
\end{theorem}

\begin{proof}
(i). The relation $\cN\preceq\cE$ follows by Lemma \ref{le:drelation} (i).

(ii). The relation $\cE \preceq \cD$ follows by Lemma \ref{le:drelation} (ii).

(iii). The equivalence relation $\cD \xlongequal{d} \cP_s \xlongequal{d} \cP$ follows by Lemma \ref{le:drelation} (iii) and the inclusion relation $\cP_s\subseteq\cP\subseteq\cD$.

(iv). The relation $\cP \succeq \cP_e$ follows directly by $\cP_e\subset\cP$.

This completes the proof.
\end{proof}

\begin{lemma}
\label{le:drelation}
(i) One can decide whether two arbitrary NPT states are LU equivalent or not, if the LU equivalence or LU inequivalence between two arbitrary EWs is decidable.

(ii) One can decide whether two arbitrary EWs are LU equivalent or not, if the LU equivalence or LU inequivalence between two arbitrary states is decidable.

(iii) One can decide whether two arbitrary states are LU equivalent or not, if and only if the LU equivalence or LU inequivalence between two arbitrary separable states is decidable.
\end{lemma}

\begin{proof}
(i) The assertion follows from the fact that the partial transpose of an NPT state is an EW. Suppose that $\rho_{AB}$ and $\sigma_{AB}$ are two NPT states. It follows that $\rho_{AB}^\Gamma$ and $\sigma_{AB}^\Gamma$ are both EWs. One can verify that the LU equivalence between two bipartite operators $M,N$ is the same as the LU equivalence between $M^\Gamma$ and $N^\Gamma$. Thus, we conclude that $\rho_{AB}\sim\sigma_{AB}$ is equivalent to $\rho_{AB}^\Gamma\sim\sigma_{AB}^\Gamma$. Due to the assumption, the LU equivalence or LU inequivalence between $\rho_{AB}^\Gamma$ and $\sigma_{AB}^\Gamma$ is decidable. Thus, the LU equivalence or LU inequivalence between $\rho_{AB}$ and $\sigma_{AB}$ is also decidable. This assertion implies that $\cN\preceq\cE$.

(ii) One may transform the problem of determining whether two EWs are LU equivalent to the problem for two states, by the following relation between EWs and states: for an arbitrary EW $W$ supported on the bipartite Hilbert space $\bbC^m\ox\bbC^n$, there exists a large enough coefficient $x$ such that 
\beq
\label{eq:ewtostate-1}
\rho_x:=W+x\I_{mn}
\eeq
is positive semidefinite, thus a non-normalized state. Let $W_1$ and $W_2$ be two arbitrary EWs supported on $\bbC^m\ox\bbC^n$. Using the technique given by Eq. \eqref{eq:ewtostate-1}, one can find two large enough coefficients $x_1,x_2$ such that $\rho_{x_1}:=W_1+x_1 \I_{mn}$ and $\rho_{x_2}:=W_2+x_2 \I_{mn}$ are two states. Then one can verify that the LU equivalence of $W_1$ and $W_2$ is simultaneous with that of $\rho_{x_1}$ and $\rho_{x_2}$. Thus, if one can decide whether two states are LU equivalent or not, then whether two EWs are LU equivalent or not is also decidable by the transformation as Eq. \eqref{eq:ewtostate-1}. This assertion implies that $\cE\preceq\cD$.

(iii) 
The ``Only if'' part follows directly by the relation $\cP_s\subset\cD$ and implies that $\cP_s\preceq\cD$. Next, we prove the ``If'' part which implies that $\cP_s\succeq\cD$. Let $\rho$ and $\sigma$ be two arbitrary states supported on the same bipartite Hilbert space. By Lemma \ref{le:noise-sep} we conclude that there exist two large enough coefficients $x_1,x_2$ such that $\rho+x_1\I$ and $\sigma+x_2\I$ are both separable. Denote $x:=\max{x_1,x_2}$. Thus, we have $\rho_x\equiv\rho+x\I$ and $\sigma_x\equiv\sigma+x\I$ are both separable. Similar to assertion (i), we know that the LU equivalence of $\rho$ and $\sigma$ is simultaneous with that of $\rho_x$ and $\sigma_x$. Due to the assumption, the LU equivalence or inequivalence of $\rho_x$ and $\sigma_x$ is decidable, and thus the LU equivalence or inequivalence of $\rho$ and $\sigma$ is also decidable. To sum up, we obtain that $\cP_s\xlongequal{d}\cD$ by definition.

This completes the proof.
\end{proof}

To prove Lemma \ref{le:drelation}, it is essential to observe the fact that for any two bipartite Hermitian operators $\rho,\sigma$ supported on $\bbC^m\ox\mathbb{C}^n$ and any real number $x$, the LU equivalence between $\rho$ and $\sigma$ is simultaneous with that between $\rho+x\I_{mn}$ and $\sigma+x\I_{mn}$. In the following lemma, we propose a set $\cD_{\lambda_1}\subset\cD$, and reveal the relation of $\cD_{\lambda_1}$ and $\cE$ on LU decidability, by virtue of the above observation.

\begin{lemma}
\label{le:fullrank}
Denote by $\cD_{\lambda_1}$ the set of states each of whose $\lambda_1$-eigenspaces contains no product vector, where $\lambda_1$ is the largest eigenvalue of a state. Then the LU decidabilities of $\cD_{\lambda_1}$ and $\cE$ admit the relation as $\cD_{\lambda_1}\preceq \cE$.
\end{lemma}

\begin{proof}
For any two given states $\rho_1,\rho_2\in\cD_{\lambda_1}$, their spectral decompositions read as $\rho_1=\sum_j\lambda_j\ketbra{\psi_j}$ and $\rho_2=\sum_j\mu_j\ketbra{\phi_j}$, where $\lambda_j$'s and $\mu_j's$ are all greater than zero, and sorted by descending order respectively. By Lemma \ref{le:spectral-dec}, one can decide $\rho_1$ amd $\rho_2$ are not LU equivalent, if there exists some $j$ such that $\lambda_j\neq \mu_j$. We thus consider the case $\lambda_j=\mu_j$ for each $j$. Based on $\rho_1,\rho_2$, we construct two operators by
\beq
\label{eq:espace-1}
\bal
W_1&:=\mu\I - \rho_1, \\
W_2&:=\mu\I - \rho_2,
\eal
\eeq
where $\mu$ is defined as
\beq
\label{eq:espace-2}
\mu:=\max_{\norm{\ket{a,b}}=1}\left\{\bra{a,b}\rho_1\ket{a,b},\bra{a,b}\rho_2\ket{a,b}\right\}.
\eeq
It follows from $\rho_1,\rho_2\in\cD_{\lambda_1}$ that the $\lambda_1$-eigenspaces of $\rho_1$ and $\rho_2$ respectively contain no product vector. Due to this fact, one can verify that $\mu$ is strictly smaller than $\lambda_1$ which is the largest eigenvalue for both $\rho_1$ and $\rho_2$. Hence, $W_1,W_2$ given by Eq. \eqref{eq:espace-1} are non-positive semidefinite. We further obtain that for any product vector $\ket{a,b}$ with norm $1$,
\beq
\label{eq:espace-3}
\bal
\bra{a,b}W_1\ket{a,b}&=\mu-\bra{a,b}\rho_1\ket{a,b}\geq 0, \\
\bra{a,b}W_2\ket{a,b}&=\mu-\bra{a,b}\rho_2\ket{a,b}\geq 0,
\eal
\eeq
where the two inequalities follow directly by Eq. \eqref{eq:espace-2}. Thus, we conclude that both $W_1,W_2$ are EWs by definition. According to Eq. \eqref{eq:espace-1}, the LU equivalence or inequivalence between $\rho_1$ and $\rho_2$ is decidable, if one can decide whether $W_1$ and $W_2$ are LU equivalent or not. Due to $W_1,W_2\in\cE$, we finally derive that $\cD_{\lambda_1}\preceq \cE$.

This completes the proof.
\end{proof}

Next, we generalize the idea in Lemma \ref{le:fullrank}, and derive a similar result on each set $\cD_{\lambda_j}$, where $\cD_{\lambda_j}$ represents the set of states each of whose $\lambda_j$-eigenspaces contains no product vector. 

\begin{lemma}
\label{le:range-2}
Denote by $\cD_{\lambda_j}$ the set of states each of whose $\lambda_j$-eigenspaces contains no product vector. Then the LU decidabilities of $\cD_{\lambda_j}$ and $\cE$ admit $\cD_{\lambda_j}\preceq\cE$.
\end{lemma}

\begin{proof}
For any two states $\rho_1,\rho_2$, one can decide that the two are not LU equivalent by Lemma \ref{le:spectral-dec}, if the eigenvalues of $\rho_1$ are not exactly the same as those of $\rho_2$. We thus consider whether there exists the LU equivalence between $\rho_1,\rho_2\in\cD_{\lambda_j}$ with identical eigenvalues. In this case, let the spectral decompositions of $\rho_1$ and $\rho_2$ respectively be:
\beq
\label{eq:range-2.1}
\bal
\rho_1&=\lambda_jP_j+\sum_{i\neq j}\lambda_i P_i, \\
\rho_2&=\lambda_jQ_j+\sum_{i\neq j}\lambda_i Q_i,
\eal
\eeq
where $P_i$'s (including $P_j$) and $Q_i$'s (including $Q_j$) are two families of orthogonal projectors. By the definition of $\cD_{\lambda_j}$, we obtain that the ranges of $P_j$ and $Q_j$ respectively contain no product vector. One can verify that for any number $x>0$, the two non-normalized states $\rho_1+x\I$ and $\rho_2+x\I$ are still in $\cD_{\lambda_j}$, and they are of full rank. Recall that the LU equivalence between $\rho_1$ and $\rho_2$ is simultaneous with that between $\rho_1+x\I$ and $\rho_2+x\I$. In light of this, it suffices to consider the states of full rank in $\cD_{\lambda_j}$, and thus each eigenvalue written in Eq. \eqref{eq:range-2.1} can be assumed as positive. 

Based on Eq. \eqref{eq:range-2.1}, we construct two Hermitian operators by
\beq
\label{eq:range-2.2}
\bal
W_1&:=-\mu P_j+\sum_{i\neq j}\lambda_i P_i, \\
W_2&:=-\mu Q_j+\sum_{i\neq j}\lambda_i Q_i,
\eal
\eeq
where $\mu>0$. We claim that there exists a small enough $\mu$ such that $W_1$ and $W_2$ are both non-normalized EWs. First, from Eq. \eqref{eq:range-2.2} both $W_1,W_2$ are not positive semidefinite. Second, to make $W_1,W_2$ are EWs, we obtain that for any product vector $\ket{a,b}$ with norm $1$,
\begin{eqnarray}
\label{eq:range-2.3.1}
&&\bra{a,b}\big(\sum_{i\neq j}\lambda_i P_i\big)\ket{a,b}-\mu\bra{a,b}P_j\ket{a,b}\geq 0, \\
\label{eq:range-2.3.2}
&&\bra{a,b}\big(\sum_{i\neq j}\lambda_i Q_i\big)\ket{a,b}-\mu\bra{a,b}Q_j\ket{a,b}\geq 0.
\end{eqnarray}
For convenience, we note that
\beq
\label{eq:range-2.4}
\bal
p_{\max}&:=\max_{\norm{\ket{a,b}}=1} \bra{a,b}P_j\ket{a,b}, \\
q_{\max}&:=\max_{\norm{\ket{a,b}}=1} \bra{a,b}Q_j\ket{a,b}.
\eal
\eeq
If $p_{\max}=0$, then \eqref{eq:range-2.3.1} holds for any $\mu>0$, and similarly  \eqref{eq:range-2.3.2} holds for any $\mu>0$ if $q_{\max}=0$. We then consider that both $p_{\max},q_{\max}$ are positive. We also note that 
\beq
\label{eq:range-2.5}
\bal
p_{\min}&:=\min_{\norm{\ket{a,b}}=1} \bra{a,b}\big(\sum_{i\neq j}\lambda_i P_i\big)\ket{a,b}, \\
q_{\min}&:=\min_{\norm{\ket{a,b}}=1} \bra{a,b}\big(\sum_{i\neq j}\lambda_i Q_i\big)\ket{a,b}.
\eal
\eeq
In fact the following map: 
\beq
\label{eq:range-2.5.1}
\ket{x}\mapsto \bra{x}\hat{P}\ket{x},
\eeq
where $\hat{P}$ is positive semidefinite and sub-normalized (i.e. $\tr(\hat{P})\leq 1$), is a continuous map, because it maps $\ket{x}$ to the square of the norm of $\sqrt{\hat{P}}\ket{x}$ which is bounded by the square of $\norm{\sqrt{\hat{P}}}\cdot\norm{\ket{x}}$. Then, by the knowledge from real analysis \cite{ranalysis}, the two minimums $p_{\min},q_{\min}$ given by Eq \eqref{eq:range-2.5} can be attained. That is, there exist some $\ket{a_1,b_1},\ket{a_2,b_2}$ such that $p_{\min}=\bra{a_1b_1}\big(\sum_{i\neq j}\lambda_i P_i\big)\ket{a_1,b_1}$ and $q_{\min}=\bra{a_2,b_2}\big(\sum_{i\neq j}\lambda_i Q_i\big)\ket{a_2,b_2}$.

We next show by contradiction that for any product vector $\ket{a,b}$, $\bra{a,b}\big(\sum_{i\neq j}\lambda_i P_i\big)\ket{a,b}>0$. It further implies that $p_{\min}>0$. One can similarly show $q_{\min}>0$. Assume that if there is some $\ket{a,b}$ such that $\bra{a,b}\big(\sum\limits_{i\neq j}\lambda_i P_i\big)\ket{a,b}=0$. It follows that $\bra{a,b}P_i\ket{a,b}=0$ for each $i\neq j$, as each eigenvalue $\lambda_i$ is positive and $P_i$'s are orthogonal projectors. It means that $\ket{a,b}$ lies in the orthogonal complement space of $\oplus_{i\neq j}\cR_i$, where $\cR_i$ denotes the range of $P_i$ for each $i$. Recall that $\rho_1,\rho_2$ are both of full rank. By Eq. \eqref{eq:range-2.1} we conclude that the orthogonal complement space of $\oplus_{i\neq j}\cR_i$ is indeed the range of $P_j$, namely $\cR_j$. Thus, $\ket{a,b}$ lies in $\cR_j$. It contradicts that $\cR_j$ contains no product vector by the definition of $\cD_{\lambda_j}$. Therefore, we conclude that both $p_{\min}>0$ and $q_{\min}>0$. Let 
\beq
\label{eq:range-2.6}
0<\mu\leq\min\{\frac{p_{\min}}{p_{\max}},\frac{q_{\min}}{q_{\max}}\}.
\eeq
One can verify that both \eqref{eq:range-2.3.1} and \eqref{eq:range-2.3.2} hold if $\mu$ satisfies \eqref{eq:range-2.6}. By the definition of EW, we obtain that $W_1$ and $W_2$ formulated by Eq. \eqref{eq:range-2.2} are both EWs, when $\mu$ is given by \eqref{eq:range-2.6}.

Finally we show that the LU equivalence between $\rho_1$ and $\rho_2$ is decidable, if that between $W_1$ and $W_2$ given by Eq. \eqref{eq:range-2.2} is decidable. By Lemma \ref{le:spectral-dec} we conclude that $W_1\sim W_2$ if and only if $P_i\sim Q_i$ by the same LU operator for each $i$. Since $\rho_1$ has the same projector as $W_1$, and $\rho_2$ has the same projector as $W_2$, then the LU equivalence between $\rho_1$ and $\rho_2$ is simultaneous with that between $W_1$ and $W_2$. If $\cE$ is LU decidable, which means that the LU equivalence between $W_1$ and $W_2$ given by Eq. \eqref{eq:range-2.2} is decidable, thus the LU equivalence between $\rho_1$ and $\rho_2$ is decidable. That is, $\cD_{\lambda_j}\preceq\cE$.

This completes the proof.
\end{proof}

For convenience, denote by $\cD_{\lambda}$ the set of states each of which satisfies that there is a $\lambda_j$-eigenspace containing no product vector. For $\cD_{\lambda}$, it is a union of all $\cD_{\lambda_j}$, and for each $\cD_{\lambda_j}$, we conclude that $\cD_{\lambda_j}\preceq\cE$ by Lemma \ref{le:range-2}. Then a corollary follows directly that $\cD_{\lambda}\preceq\cE$.
One can verify that if each eigenspace of a state $\rho$ contains at least one product vector, then for any local unitary operator $U\ox V$, each eigenspace of $(U\ox V)\rho(U\ox V)^\dg$ also contains at least one product vector. It implies that two states are LU inequivalent, if one belongs to $\cD_{\lambda}$ and the other does not belong to $\cD_{\lambda}$. 

We further show the inverse part of the LU decidability relation between $\cD_{\lambda}$ and $\cE$, namely $\cE\preceq\cD_{\lambda}$.

\begin{lemma}
\label{le:range-3}
Denote by $\cD_{\lambda}$ the set of states satisfying that at least one eigenspace contains no product vector. Then the LU decidabilities of $\cD_{\lambda}$ and $\cE$ admit $\cE\preceq\cD_{\lambda}$.
\end{lemma}

\begin{proof}
For any two given EWs $W_1,W_2\in\cE$, their spectral decompositions read as 
\beq
\label{eq:ewequiv-1}
W_1=\sum_{j=1}^k \lambda_j P_j, \quad
W_2=\sum_{j=1}^k \lambda_j Q_j,
\eeq
where $\lambda_j$'s are eigenvalues in descending order. Here we may assume $W_1$ and $W_2$ share the identical eigenvalues. Otherwise, one can decide that $W_1$ and $W_2$ cannot be LU equivalent by Lemma \ref{le:spectral-dec}. Since $W_1$ and $W_2$ are EWs, there is at least one negative eigenvalue for $W_1$ and $W_2$ respectively. By virtue of the definition of EWs, for each negative eigenvalue $\lambda_x$, the corresponding eigenspace is a completely entangled subspace, i.e. contains no product vector. Otherwise, if there is a product vector $\ket{a,b}$ in $\lambda_x$-eigenspace of $W_1$ for some $\lambda_x<0$, then one can verify that $\bra{a,b}W_1\ket{a,b}<0$ which contradicts that $W_1$ is an EW. Based on Eq. \eqref{eq:ewequiv-1} we propose two states as follows:
\beq
\label{eq:ewequiv-2}
\rho_1=\sum_{j=1}^k \mu_j P_j, \quad
\rho_2=\sum_{j=1}^k \mu_j Q_j,
\eeq
where each $\mu_j$ is positive. According to the analysis above, we know that $\rho_1$ and $\rho_2$ both have one eigenspace including no product vector, and thus both $\rho_1$ and $\rho_2$ are in $\cD_{\lambda}$. Therefore, if $\cD_{\lambda}$ is LU decidable, the LU equivalence between $\rho_1$ and $\rho_2$ is clear, and correspondingly the LU equivalence between $W_1$ and $W_2$ given by Eq. \eqref{eq:ewequiv-1} can be determined. In other words, we obtain $\cE\preceq\cD_{\lambda}$. This completes the proof.
\end{proof}

Based on the two lemmas above, we establish an equivalence relation on LU decidability between two sets of different types of Hermitian operators, by the following theorem. Theorem \ref{thm:mainrelation} shows that the set of EWs $\cE$ is LU decidable, if and only if a subset of bipartite states $\cD_{\lambda}$ is LU decidable.

\begin{theorem}
\label{thm:mainrelation}
The LU decidability of $\cE$ is equal to that of $\cD_{\lambda}$, namely $\cE\xlongequal{d}\cD_{\lambda}$.
\end{theorem}

Theorem \ref{thm:mainrelation} follows directly by combining Lemmas \ref{le:range-2} and \ref{le:range-3} whose proofs are given above. By Theorem \ref{cor:drelation} and Theorem \ref{thm:mainrelation}, we establish the relations on LU decidabilities among the set of EWs $\cE$ and several sets of bipartite states. These known relations present a two-level hierarchy depicted in Fig. \ref{fig:1}.

The set $\cD_\lambda$ intersects with $\cN,\cP_e,\cP_s$ respectively. The following remark roughly characterizes the states in $\cD_{\lambda}$.
\begin{remark}
\label{re:cdl}
The set $\cD_\lambda$ contains NPT states, PPT entangled states and separable states. We construct examples to show the existence of such three classes of states in $\cD_{\lambda}$. First, we propose an NPT state in $\cD_\lambda$ as follows:
\beq
\label{eq:cdl-1}
\rho_1:=\frac{1}{10}\left(\ketbra{01}+\ketbra{10}+2\ketbra{\phi_+}+6\ketbra{\phi_-}\right),
\eeq
where $\ket{\phi_+}=\frac{1}{\sqrt{2}}(\ket{00}+\ket{11})$ and $\ket{\phi_-}=\frac{1}{\sqrt{2}}(\ket{00}-\ket{11})$. One can verify that $\rho_1$ given by Eq. \eqref{eq:cdl-1} belongs to $\cD_\lambda$ as Eq. \eqref{eq:cdl-1} is the spectral decomposition of $\rho_1$. According to Eq. \eqref{eq:cdl-1}, the eigenspaces respectively corresponding to the eigenvalues $1/5,3/5$ contain no product vector. Further, by calculation the four eigenvalues of $\rho_1^\Gamma$ are $0.4,0.4,0.3,-0.1$, and thus $\rho_1$ is an NPT state. 

Second, we propose a PPT entangled state in $\cD_\lambda$. Recall the construction of PPT entangled states by virtue of UPBs. Denote by $\cB=\{\ket{\ket{a_i,b_i}}\}_{i=1}^l$ a UPB. Then we may construct a PPT entangled state as
\beq
\label{eq:ppte-upb}
\sigma=\frac{1}{n-l}\left(\I_n-\sum_{i=1}^l \ketbra{a_i,b_i}\right).
\eeq
By the definition of UPB we know the range of $\sigma$ is a completely entangled subspace, and thus $\sigma$ is entangled. It follows that there exists a small enough $\epsilon>0$ such that $\sigma+\epsilon\I$ is still PPT entangled and has full rank. Thus we obtain that a PPT entangled state $\rho_2:=\sigma+\epsilon\I\in\cD_\lambda$. 

Finally we construct a separable state in $\cD_\lambda$. The following is a two-qubit separable state of rank three:
\beq
\label{eq:cdl-3}
\rho'_3:=\frac{1}{6}\left(\ketbra{00}+\ketbra{11}+4\ketbra{+,+}\right),
\eeq
where $\ket{+}=\frac{1}{\sqrt{2}}(\ket{0}+\ket{1})$. One may check that $\ket{01}-\ket{10}$ is orthogonal to the range of $\rho'_3$. Thus, we conclude that $\rho_3:=\rho'_3+\epsilon\I$ for some positive $\epsilon$ has an eigenspace spanned by $\ket{01}-\ket{10}$ only. It means that $\rho_3$ is a separable state and has an eigenspace including no product vector. Hence, $\rho_3$ belongs to $\cD_\lambda$.
\qed
\end{remark}



We then extend our study to a larger set including $\cD_{\lambda}$, denoted by $\overline{\cD_\lambda}$. The state in $\overline{\cD_\lambda}$ satisfies that there exists some eigenspace not spanned by product vectors. By Proposition \ref{prop:kprodv} below we show that the LU decidability of $\overline{\cD_\lambda}$ is equal to that of $\cD_{\lambda}$, namely $\overline{\cD_\lambda}\xlongequal{d}\cD_{\lambda}$.

\begin{proposition}
\label{prop:kprodv}
Denote by $\overline{\cD_\lambda}$ the set of states each of which satisfies that there exists some eigenspace not spanned by product vectors. Then $\overline{\cD_\lambda}$ is LU decidable, if and only if $\cD_{\lambda}$ is LU decidable. That is, $\overline{\cD_{\lambda}}\xlongequal{d}\cD_{\lambda}$.
\end{proposition}

\begin{proof}
First, we conclude that $\cD_{\lambda}\preceq\overline{\cD_{\lambda}}$ due to $\cD_{\lambda}\subset\overline{\cD_{\lambda}}$. For $\overline{\cD_{\lambda}}\xlongequal{d}\cD_{\lambda}$ to hold, we next have to show $\overline{\cD_{\lambda}}\preceq\cD_{\lambda}$. It means to show that for any two states in $\overline{\cD_{\lambda}}$, whether they are LU equivalent or not is decidable, if the set $\cD_{\lambda}$ is LU decidable.
Suppose that $\rho,\sigma$ are two arbitrary states in $\overline{\cD_{\lambda}}$. It suffices to consider that $\rho,\sigma$ have identical eigenvalues, otherwise the two states are LU inequivalent. We may also assume that $\rho,\sigma$ both have full rank by means of adding the same white noise, namely $\rho+x\I$ and $\sigma+x\I$. 

Let the spectral decompositions of $\rho$ and $\sigma$ be:
\beq 
\label{eq:kprodv-1}
\bal
\rho&=\sum_{j=1}^l \lambda_j P_j, \\
\sigma&=\sum_{j=1}^l \lambda_j Q_j, \\
\eal
\eeq
where $\lambda_j$'s are all positive, and for each $j$ the projectors $P_j$ and $Q_j$ have the same rank. By the definition of $\overline{\cD_{\lambda}}$, it means that for each of $\rho$ and $\sigma$, there exists some eigenspace not spanned by product vectors. For the case when $\cR(P_j)$ is not spanned by product vectors while $\cR(Q_j)$ is spanned by product vectors, one can verify that $P_j\nsim Q_j$, and thus conclude that $\rho\nsim\sigma$ by Lemma \ref{le:spectral-dec}. Then it suffices to consider that for some $j$, both $\cR(P_j)$ and $\cR(Q_j)$ are not spanned by product vectors. Without loss of generality we may assume that both $\cR(P_1)$ and $\cR(Q_1)$ are not spanned by product vectors. We may decompose $P_1$ and $Q_1$ as below:
\beq
\label{eq:rprod-2}
\bal
P_1&=P_{11}+P_{12} \\
Q_1&=Q_{11}+Q_{12}, \\
\eal
\eeq
where $P_{11},P_{12}$ are orthogonal projectors, $Q_{11},Q_{12}$ are orthogonal projectors, and $\cR(P_{11})$ is spanned by product vectors and $\cR(P_{12})$ contains no product vector, $\cR(Q_{11})$ is spanned by product vectors and $\cR(Q_{12})$ contains no product vector. In other words, all product vectors in $\cR(P_1)$ are included in $\cR(P_{12})$, the same for $\cQ(Q_1)$. Substituting the decompositions into Eq. \eqref{eq:kprodv-1} we obtain that 
\beq
\label{eq:kprodv-3}
\bal
\rho&=\lambda_1 P_{11} + \lambda_1 P_{12}+\sum_{j=2}^l \lambda_j P_j, \\
\sigma&=\lambda_1 Q_{11} + \lambda_1 Q_{12}+\sum_{j=2}^l \lambda_j Q_j. \\
\eal
\eeq

We claim that $\rho\sim\sigma$ if and only if $(P_{11},P_{12},P_2,\cdots P_l)$ is SLU equivalent to $(Q_{11},Q_{12},Q_2,\cdots Q_l)$ for the following reason. The ``If'' part can be verified directly. Then we show the ``Only if'' part. With the condition $\rho\sim\sigma$, we obtain that $(P_1,P_2,\cdots P_l)$ is SLU equivalent to $(Q_1,Q_2,\cdots Q_l)$. It follows that there exists an LU operator $U\ox V$ such that
\begin{eqnarray}
\label{eq:kprodv-4.1}
&&(U\ox V)(P_{11}+P_{12})(U\ox V)^\dg=Q_{11}+Q_{12}, \\
\label{eq:kprodv-4.2}
&&(U\ox V)P_j(U\ox V)^\dg=Q_j,~~j>1.
\end{eqnarray}
One can verify that the range of $(U\ox V)P_{11}(U\ox V)^\dg$ is spanned by product vectors, and the range of $(U\ox V)P_{12}(U\ox V)^\dg$ contains no product vector. By Eq. \eqref{eq:kprodv-4.1} we conclude that the range of $(U\ox V)P_{11}(U\ox V)^\dg$ is the same as $\cR(Q_{11})$, and the range of $(U\ox V)P_{12}(U\ox V)^\dg$ is the same as $\cR(Q_{12})$. It is known that two projectors are the same if they have the same range. Thus, we derive that $(U\ox V)P_{11}(U\ox V)^\dg=Q_{11}$, and $(U\ox V)P_{12}(U\ox V)^\dg=Q_{12}$. Combining this conclusion with Eq. \eqref{eq:kprodv-4.2} we obtain that $(P_{11},P_{12},P_2,\cdots P_l)$ is SLU equivalent to $(Q_{11},Q_{12},Q_2,\cdots Q_l)$, and thus the ``Only if'' part holds.

By adjusting the eigenvalues of $\rho,\sigma$ we may construct two states in $\cD_{\lambda}$ as:
\beq
\label{eq:kprodv-5}
\bal
\rho'&=\mu_{11} P_{11}+\mu_{12} P_{12}+\sum_{j=2}^l\mu_j P_j, \\
\sigma'&=\mu_{11} Q_{11}+\mu_{12} Q_{12}+\sum_{j=2}^l\mu_j Q_j,
\eal
\eeq
where the coefficients $\mu_{11},\mu_{12},\mu_2,\cdots\mu_l$ are all positive. Because $P_{11}$ is orthogonal to $P_{12}$ and $Q_{11}$ is orthogonal to $Q_{12}$, the decompositions given by Eq. \eqref{eq:kprodv-5} are spectral decompositions of $\rho'$ and $\sigma'$. It is known that $\rho',\sigma'$ are included in $\cD_{\lambda}$, as $\cR(P_{12})$ and $\cR(Q_{12})$ both contain no product vector. By Lemma \ref{le:spectral-dec} we conclude that the LU equivalence between $\rho$ and $\sigma$ is simultaneous with that between $\rho'$ and $\sigma'$. The LU equivalence between $\rho'$ and $\sigma'$ can be decided when $\cD_{\lambda}$ is LU decidable. Therefore, we have shown that whether $\rho$ and $\sigma$ are LU equivalent or not is decidable if $\cD_{\lambda}$ is LU decidable, i.e. $\overline{\cD_{\lambda}}\preceq\cD_{\lambda}$.

To sum up, the relation $\overline{\cD_{\lambda}}\xlongequal{d}\cD_{\lambda}$ holds. This completes the proof.
\end{proof}

Finally, we show two special cases where the LU equivalence between two normalized states can be determined by specific eigenvalues of the partial transposes of their density matrices, as follows. It is known from Ref. \cite{Neigenv2013} that for any normalized bipartite states with trace $1$, all eigenvalues of the partial transpose of its density matrix lie within $[-1/2,1]$. Here, we consider two extreme cases of maximum and minimum eigenvalues, and respectively identify all states whose partial transposes have the maximum eigenvalue $1$ and the minimum eigenvalue $-1/2$. 

\begin{proposition}
\label{le:eigrange}
For a normalized bipartite state, the minimum eigenvalue of its partial transpose saturates the lower bound $-\frac{1}{2}$ if and only if the corresponding state is a two-qubit maximally entangled state; the maximum eigenvalue of its partial transpose saturates the upper bound $1$ if and only if the corresponding state is a pure product state.
\end{proposition}

The proof of Proposition \ref{le:eigrange} is provided in Appendix \ref{sec:proof1}. Since all two-qubit maximally entangled states are LU equivalent, and all pure product states are LU equivalent, we directly know that the states whose partial transposes share the eigenvalue $-1/2$ are all LU equivalent, and the the states whose partial transposes share the eigenvalue $1$ are all LU equivalent, by Proposition \ref{le:eigrange}.

\section{SLU equivalence between tuples of bipartite projectors}
\label{sec:rangelu}

In this section we consider how to determine whether two tuples of projectors are SLU equivalent or not. From Lemma \ref{le:spectral-dec} we realize that the SLU equivalence between tuples of mutually orthogonal projectors is essential for LU equivalent Hermitian operators, by analyzing their spectral decompositions. By Example \ref{ex:slu-c}, we show that for two tuples each of which contains more than two projectors, the partial SLU equivalence between parts of such two tuples is necessary but not sufficient to ensure the overall SLU equivalence between the two entire tuples. In Theorem \ref{le:slu}, we propose a necessary and sufficient condition on bipartite states in terms of the twirling operation, such that two tuples of states are SLU equivalent. 

To further emphasize the key role of SLU equivalence, we also consider the SLU equivalence between two tuples of subspaces from the ranges of projectors, i.e. the eigenspaces.
Considering two bipartite states $\rho\sim\sigma$, it follows directly that their ranges admit the relation $\cR(\rho)=(U\ox V)\cR(\sigma)$ for some local unitary $U\ox V$. In other words, if there exists a $\ket{\alpha}\in\cR(\rho)$ such that $(U\ox V)\ket{\alpha}\notin\cR(\sigma)$ for any local unitary $U\ox V$, one can decide $\rho\nsim\sigma$. Conversely, the relation $\cR(\rho)=(U\ox V)\cR(\sigma)$ is not sufficient to ensure $\rho\sim\sigma$. We propose a pair of states as below:
\beq
\label{eq:crlu-1}
\bal
\rho&:=\frac{1}{2}\ketbra{00}+\frac{1}{4}\ketbra{01}+\frac{1}{8}(\ketbra{10}+\ketbra{11}), \\
\sigma&:=\frac{1}{2}\ketbra{\phi_+}+\frac{1}{4}\ketbra{\phi_-}+\frac{1}{8}(\ketbra{01}+\ketbra{10}),
\eal
\eeq
where $\ket{\phi_+},\ket{\phi_-}$ are a pair of Bell states.
Since $\rho,\sigma$ both have full rank, their ranges are both $\bbC^2\ox\bbC^2$, and thus the relation $\cR(\rho)=(U\ox V)\cR(\sigma)$ is naturally satisfied. However, the eigenvector $\ket{00}$ of $\rho$ is not LU equivalent to the eigenvector $\ket{\phi_+}$ of $\sigma$, corresponding to the same eigenvalue $\frac{1}{2}$. By Lemma \ref{le:spectral-dec} we conclude that $\rho$ and $\sigma$ are LU inequivalent. 
It motivates us to further consider how to determine the LU equivalence between two bipartite states by their eigenspaces. A determination criterion for $\rho\sim\sigma$ is presented as below.

\begin{lemma}
\label{le:crlu-1}
Suppose that two bipartite states $\rho,\sigma$ have spectral decompositions as $\rho=\sum_{j=1}^r \lambda_j P_j$ and $\sigma=\sum_{j=1}^r \lambda_j Q_j$. Then $\rho\sim\sigma$, if and only if for each $j$, $\cR(P_j)=(U\ox V)\cR(Q_j)$ by a common LU operator $U\ox V$.
\end{lemma}

We shall show the detailed proof of Lemma \ref{le:crlu-1} in Appendix \ref{sec:proof2}.

In light of the importance of SLU equivalence, we next deeply investigate the conditions on the SLU equivalence relation between two tuples of mutually orthogonal projectors. By virtue of the Schmidt decomposition, such conditions are effective to check the LU equivalence between two arbitrary bipartite Hermitian operators. For two binary tuples of orthogonal projectors $(P_1,P_2)$ and $(Q_1,Q_2)$, the conditions both $P_1\sim P_2$ and $Q_1\sim Q_2$ are generally not sufficient for the SLU equivalence relation between the two binary tuples $(P_1,P_2)$ and $(Q_1,Q_2)$. One may refer to the counterexample formulated in the paragraph below Lemma \ref{le:spectral-dec}. 

We further consider the cases where the tuples contain more than two mutually orthogonal projectors, by focusing on the tuples of three projectors. The cases where the tuples contain more than three mutually orthogonal projectors can be similarly investigated. Suppose that $(P_1,P_2,P_3)$ and $(Q_1,Q_2,Q_3)$ are two tuples of mutually orthogonal projectors, namely $\tr(P_iP_j)=\tr(Q_iQ_j)=\delta_{ij}$. In order to make $(P_1,P_2,P_3)\sim_s (Q_1,Q_2,Q_3)$, we shall assign a necessary precondition that $(P_1,P_2)\sim_s(Q_1,Q_2)$, $(P_1,P_3)\sim_s(Q_1,Q_3)$, and $(P_2,P_3)\sim_s(Q_2,Q_3)$. Otherwise, if one of such three SLU equivalence relations is violated, one can conclude that $(P_1,P_2,P_3)$ is not SLU equivalent to $(Q_1,Q_2,Q_3)$. Then we keep studying whether this precondition is sufficient to ensure $(P_1,P_2,P_3)\sim_s (Q_1,Q_2,Q_3)$ or not.


By the following lemma we may simply consider the two tuples with $P_1=Q_1$ and $P_2=Q_2$. Then the number of projectors to be processed is reduced from six to four.

\begin{lemma}
\label{le:4prjs}
Suppose that $(P_1,P_2,P_3)$ and $(Q_1,Q_2,Q_3)$ are two tuples of mutually orthogonal projectors, which satisfy that $(P_1,P_2)\sim_s(Q_1,Q_2)$, $(P_1,P_3)\sim_s(Q_1,Q_3)$, and $(P_2,P_3)\sim_s(Q_2,Q_3)$. By simultaneously applying any LU operations to $Q_i$'s, namely $\tilde{Q}_i:=(U\ox V)Q_i(U\ox V)^\dg$ for any LU operator $(U\ox V)$, then

(i) the SLU equivalences between $(P_i,P_j)$ and $(\tilde{Q}_i,\tilde{Q}_j)$, for any $1\leq i<j\leq 3$, are all maintained;

(ii) $(P_1,P_2,P_3)$ is SLU equivalent to $(Q_1,Q_2,Q_3)$, if and only if $(P_1,P_2,P_3)$ is SLU equivalent to $ (\tilde{Q}_1,\tilde{Q}_2,\tilde{Q}_3)$.
\end{lemma}

We shall prove Lemma \ref{le:4prjs} in Appendix \ref{sec:proof2}. Recall the precondition that $(P_1,P_2)\sim_s(Q_1,Q_2)$, which means there exists an LU operator $U\ox V$, such that $P_1=(U\ox V)Q_1(U\ox V)^\dg$ and $P_2=(U\ox V)Q_2(U\ox V)^\dg$. According to Lemma \ref{le:4prjs}, with the same LU operator $U\ox V$ applying to all $Q_i$'s, we consequently obtain that $(P_1,P_2,P_3)\sim_s(Q_1,Q_2,Q_3)$ if and only if $(P_1,P_2,P_3)\sim_s(P_1,P_2,\tilde{Q}_3)$. Then we may equivalently examine whether there exists the SLU equivalence between $(P_1,P_2,P_3)$ and $(Q_1,Q_2,Q_3)$ with $P_1=Q_1$ and $P_2=Q_2$, under the precondition.


Here, we claim that the precondition mentioned above is not sufficient to ensure $(P_1,P_2,P_3)\sim_s (Q_1,Q_2,Q_3)$, and show our claim by proposing examples as follows.

\begin{example}
\label{ex:slu-c}
By constructing specific examples, we conclude that the two tuples of mutually orthogonal projectors, $(P_1,P_2,P_3)$ and $(Q_1,Q_2,Q_3)$, may not be SLU equivalent, even if $(P_i,P_j)\sim_s(Q_i,Q_j)$ holds for any $1\leq i<j\leq 3$. According to the analysis above, we may assume that $P_1=Q_1$ and $P_2=Q_2$. Then the four distinct projectors supported on $\bbC^4\ox\bbC^4$ are formulated as below:
\beq
\label{eq:cex-1}
\bal
P_1=Q_1&=\proj{0}\ox\diag (1,1,0,0), \\
P_2=Q_2&=\proj{1}\ox\diag (1,0,1,0), \\
P_3&=\proj{2}\ox\diag (1,0,0,1), \\
Q_3&=\proj{2}\ox\diag (0,1,1,0).
\eal
\eeq
First, due to $P_1=Q_1$ and $P_2=Q_2$, the relation $(P_1,P_2)\sim_s(Q_1,Q_2)$ naturally holds by the LU matrix $\I\ox\I$. Second, one can verify that $(P_1,P_3)$ is SLU equivalent to $(Q_1,Q_3)$ by the LU matrix as
\beq
\label{eq:cex-2}
\I\ox
\left(\bma
0 & 1 \\
1 & 0
\ema\oplus
\bma
0 & 1 \\
1 & 0
\ema\right).
\eeq
Third, one can verify that $(P_2,P_3)$ is SLU equivalent to $(Q_2,\rho_3)$ by the LU matrix as
\beq
\label{eq:cex-3}
\I\ox
\bma
0 & 0 & 1 & 0 \\
0 & 0 & 0 & 1 \\
1 & 0 & 0 & 0 \\
0 & 1 & 0 & 0 
\ema.
\eeq
So the four projectors satisfy the precondition.
\qed
\end{example}
In the following proof, we show that the two tuples given in Example \ref{ex:slu-c} cannot be SLU equivalent.
\begin{proof}
We show by contradiction that $(P_1,P_2,P_3)$ and $(Q_1,Q_2,Q_3)$ formulated by Eq. \eqref{eq:cex-1} are not SLU equivalent. Assume that $(P_1,P_2,P_3)\sim_s(Q_1,Q_2,Q_3)$, and denote by $U_1\ox U_2$ an arbitrary LU matrix such that $P_j=(U_1\ox U_2)Q_j(U_1\ox U_2)^\dg$ for each $j$. By virtue of Eq. \eqref{eq:cex-1} and observing the second subsystems of $P_1,P_2$, we obtain that 
\beq
\label{eq:cex-4}
\bal
U_2\left(\diag (1,1,0,0)\right)U_2^\dg&=\diag (1,1,0,0), \\
U_2\left(\diag (1,0,1,0)\right)U_2^\dg&=\diag (1,0,1,0).
\eal
\eeq
It means that the unitary matrix $U_2$ keeps both $\diag (1,1,0,0)$ and $\diag (1,0,1,0)$ invariant under the unitary similarity. By Lemma \ref{le:diag-inv} we conclude that $U_2$ must be diagonal, namely $\diag(e^{i\t_1},e^{i\t_2},e^{i\t_3},e^{i\t_4})$. Nevertheless, such diagonal $U_2$ cannot make that $U_2(\diag(1,0,0,1))U_2^\dg=\diag(0,1,1,0)$ holds. It means that $P_3$ cannot be locally transformed to $Q_3$ by $U_1\ox U_2$, which contradicts the assumption. Therefore, the formulated two tuples cannot be SLU equivalent. This completes the proof.
\end{proof}

Example \ref{ex:slu-c} indicates the partial SLU equivalence cannot ensure the overall SLU equivalence. It implies that the problem of determining the SLU equivalence between two arbitrary tuples of mutually orthogonal projectors could become more and more difficult as the number of projectors included in a tuple increases. 
A projector can be regarded as a non-normalized state. Thus, we may derive a more general conclusion that for two tuples of states, the partial SLU equivalence cannot ensure the overall SLU equivalence. Specifically, for two tuples of states denoted by $(\rho_1,\rho_2,\rho_3)$ and $(\sigma_1,\sigma_2,\sigma_3)$, they may not be SLU equivalent, even if $(\rho_i,\rho_j)\sim_s(\sigma_i,\sigma_j)$ holds for any $1\leq i<j\leq 3$. 

Along with this broader perspective by extending projectors to states, we study the SLU equivalence between tuples of states, and propose a necessary and sufficient condition for the SLU equivalence in Theorem \ref{le:slu}. This necessary and sufficient condition is formulated in terms of the twirling operation which has been widely used in the resource theory of asymmetry \cite{ssr2007,rfgour2008,gour2009}. It suffices to assume that $\rho_1=\sigma_1$ for the two tuples $(\rho_1,\rho_2,\rho_3)$ and $(\sigma_1,\sigma_2,\sigma_3)$ with the similar idea given by Lemma \ref{le:4prjs}.

\begin{theorem}
\label{le:slu}
Suppose that $\rho_1=\sigma_1,\rho_2=\sigma_2,\rho_3,\sigma_3$ are normalized states which satisfy that $(\rho_i,\rho_j)\sim_s(\sigma_i,\sigma_j)$ for any $1\leq i<j\leq 3$. Then the tuple $(\rho_1,\rho_2,\rho_3)$ is SLU equivalent to $(\sigma_1,\sigma_2,\sigma_3)$, if and only if there exist two interseted groups $\cU_1$ and $\cU_2$ of LU operators such that:

(i) $\rho_1=\int_{\cU_1}U\alpha U^\dg dU$ for some state $\alpha$, and $\rho_2=\int_{\cU_2}U\beta U^\dg dU$ for some state $\beta$, with the two integrals over the Haar measure;

(ii) $\rho_3=(U\ox V)\sigma_3 (U\ox V)^\dg$ for some $U\ox V\in \cU_1\cap\cU_2$.
\end{theorem}

\begin{proof}
Due to $\rho_1=\sigma_1,\rho_2=\sigma_2$, by definition one can verify that $(\rho_1,\rho_2,\rho_3)$ is SLU equivalent to $(\sigma_1,\sigma_2,\sigma_3)$ if and only if there exists a common LU operator $U\ox V$ which simultaneously commutes with $\rho_1$ and $\rho_2$, and satisfies $\rho_3=(U\ox V)\sigma_3 (U\ox V)^\dg$. Next, we show that the combination of conditions (i) and (ii) is equivalent to this necessary and sufficient condition.

First we show that the LU operators commuting with a bipartite operator form a group under the matrix multiplication. Let $\rho$ be a bipartite operator. Then the identity operator $\I\ox \I$ must commute with $\rho$. Moreover, assume that $\rho$ commutes with two LU operator $U_1\ox V_1$ and $U_2\ox V_2$, i.e., $\rho(U_1\ox V_1)=(U_1\ox V_1)\rho$ and $\rho(U_2\ox V_2)=(U_2\ox V_2)\rho$. One can directly verify that $\rho$ also commutes with the LU operator $(U_1U_2)\ox (V_1V_2)$ as
\beq
\label{eq:slu-c1}
\bal
&\quad\rho\left[(U_1U_2)\ox (V_1V_2)\right]=\rho (U_1\ox V_1) (U_2\ox V_2) \\
&=(U_1\ox V_1)\rho(U_2\ox V_2)=(U_1\ox V_1)(U_2\ox V_2)\rho\\
&=\left[(U_1U_2)\ox (V_1V_2)\right]\rho.
\eal
\eeq
Hence, the LU operators commuting with a bipartite operator form a multiplication group.

Second, for a given LU group $\cU_g$, the bipartite states commuting with $\cU_g$ is denoted by $\mathrm{inv}(\cU_g)$,
\beq
\label{eq:slu-c2}
\mathrm{inv}(\cU_g)\equiv\{ \rho\in\cD(\cH^{AB})~|~\forall~X\in\cU_g:X\rho X^\dg=\rho\}.
\eeq
It is known that the set $\mathrm{inv}(\cU_g)$ can be characterized by the useful twirling operation \cite{rfgour2008}, namely
\beq
\label{eq:slu-c3}
\cT[\rho]\equiv\int_\cU U\rho U^\dg dU,
\eeq
for some unitary group $\cU$. The twirling operation averages over the action of the unitary group $\cU$ with the Haar measure. If the unitary group $\cU$ is finite, one simply replaces the integral with a sum as
\beq
\label{eq:slu-c3.1}
\cT[\rho]\equiv\frac{1}{\abs{\cU}}\sum_{U\in\cU} U\rho U^\dg.
\eeq
By virtue of the twirling operation, we may characterize $\mathrm{inv}(\cU_g)$ as follows:
\beq
\label{eq:slu-4}
\mathrm{inv}(\cU_g)=\left\{\rho\in\cD(\cH^{AB})~|~\cG[\rho]=\rho\right\},
\eeq
where $\cG[\rho]$ is given by
\beq
\label{eq:slu-4.1}
\cG[\rho]=\int_{\cU_g} U\rho U^\dg dU.
\eeq
Since $\cG$ is idempotent, namely $\cG^2\equiv\cG\circ\cG=\cG$, the set $\mathrm{inv}(\cU_g)$ is indeed the image of $\cG$, i.e.,
\beq
\label{eq:slu-4.2}
\mathrm{inv}(\cU_g)=\left\{\rho~|~\rho=\cG[\sigma],~\forall\sigma\in\cD(\cH^{AB})\right\},
\eeq

Suppose that $\cU_1$ and $\cU_2$ are two intersected groups of LU operators. Based on the analysis above, the condition (i) ensures that $\rho_1$ commutes with $\cU_1$, and $\rho_2$ commutes with $\cU_2$. Since $\cU_1\cap \cU_2$ is non-empty, it follows that the elements in $\cU_1\cap \cU_2$ commute with $\rho_1$ and $\rho_2$ simultaneously. Additionally, combining with condition (ii), we conclude that there exists a common LU operator which simultaneously commutes with $\rho_1$ and $\rho_2$, and makes $\rho_3\sim\sigma_3$. Therefore, the addition of conditions (i) and (ii) is necessary and sufficient to $(\rho_1,\rho_2,\rho_3)\sim_s(\sigma_1,\sigma_2,\sigma_3)$ with $\rho_1=\sigma_1,\rho_2=\sigma_2$.

This completes the proof.
\end{proof}

By viewing projectors as special states, one may similarly derive a necessary and sufficient condition for the SLU equivalence between tuples of projectors.

\section{Concluding remarks}
\label{sec:con}

In this paper we introduced a relation in terms of the decidabilities of LU equivalence (also known as the LU decidabilities for short) between different sets of bipartite Hermitian operators, and connected the LU decidabilities for the set of EWs and several sets of states. This idea may shed new light on the problem of determining whether two arbitrary states are LU equivalent or not.

By virtue of the relation ``$\prec$'' defined by Definition \ref{def:decide}, we compared the LU decidability of the set of EWs with that of several sets of states. Based on the comparisons we established a hierarchy on LU decidabilities for these sets. The hierarchy has been illustrated by Fig. \ref{fig:1}. By Theorem \ref{cor:drelation} we have shown the relations of LU decidability among the sets of states classified by the PPT criterion, and the set of EWs. Then we proposed a set of states satisfying that at least one eigenspace contains no product vector, denoted by $\cD_{\lambda}$. We indicated in Theorem \ref{thm:mainrelation} that the two LU decidabilities for $\cD_{\lambda}$ and the set of EWs are equal.
We further extended our study to a larger set including $\cD_{\lambda}$, each of whose states satisfies that there is one eigenspace not spanned by product vectors. We refer to such a set as $\overline{\cD_{\lambda}}$, and have shown that the LU decidability of $\overline{\cD_{\lambda}}$ is equal to that of $\cD_{\lambda}$ by Proposition \ref{prop:kprodv}.
By analyzing spectral decompositions, we realized that the SLU equivalence is crucial to LU equivalent operators. In light of this, we studied the SLU equivalence between tuples of mutually orthogonal projectors. By Example \ref{ex:slu-c} we revealed that for two tuples of mutually orthogonal projectors, the partial SLU equivalence cannot ensure the overall SLU equivalence. For this reason, we provided a necessary and sufficient condition in Theorem \ref{le:slu}, such that two tuples of bipartite states are SLU equivalent.

There are some interesting problems remaining for future work. One direction is to complement the relations on the LU decidability for the sets appearing in Fig. \ref{fig:1}. Another direction is to further characterize the SLU equivalent tuples of mutually orthogonal projectors, in order to propose more efficient criteria of determining whether two states are LU equivalent or not.

\section*{acknowledgements}

This work is funded by the NNSF of China (Grant Nos. 12401597, 12471427), the Basic Research Program of Jiangsu (Grant No. BK20241603), and the Wuxi Science and Technology Development Fund Project (Grant No. K20231008). 

\appendix

\section{Proof of Proposition \ref{le:eigrange}}
\label{sec:proof1}

Here, we present a detailed proof of Proposition \ref{le:eigrange}. For this purpose, we first have to recall the following known conclusion in matrix theory.

\begin{lemma}\cite[Corollary 4.3.15.]{bookmatrix}
\label{le:eigenineq}
Let $A,B\in\cM_n$ be Hermitian. The eigenvalues here are arranged in algebraically nondecreasing order. Then
\beq
\label{eq:eigenineq-1}
\lambda_i(A)+\lambda_1(B)\leq\lambda_i(A+B)\leq\lambda_i(A)+\lambda_n(B), ~~i=1,\cdots,n
\eeq
with equality in the upper bound if and only if there is nonzero vector $\ket{x}$ such that 
\beq
\label{eq:eigenineq-2}
\bal
&A\ket{x}=\lambda_i(A)\ket{x},~~B\ket{x}=\lambda_n(B)\ket{x},\\
\text{and} ~~&(A+B)\ket{x}=\lambda_i(A+B)\ket{x};
\eal
\eeq
equality in the lower bound occurs if and only if there is nonzero vector $\ket{x}$ such that 
\beq
\label{eq:eigenineq-3}
\bal
&A\ket{x}=\lambda_i(A)\ket{x},~~ B\ket{x}=\lambda_1(B)\ket{x},\\
\text{and}~~ &(A+B)\ket{x}=\lambda_i(A+B)\ket{x}.
\eal
\eeq
If $A$ and $B$ have no common eigenvector, then every inequality in \eqref{eq:eigenineq-1} is a strict inequality.
\end{lemma}

Now we are ready to show Proposition \ref{le:eigrange} in detail. Note that Proposition \ref{le:eigrange} is a more refined conclusion based on \cite[Theorem 2.]{Neigenv2013}, where we completely identify the two extreme cases of minimum and maximum eigenvalues. For the proof of Proposition \ref{le:eigrange} to be self-consistent, we shall follow the proof of \cite[Theorem 2.]{Neigenv2013} to complete our proof.

\textbf{Proof of Proposition \ref{le:eigrange}.}
We first consider a pure state $\ket{\psi}$ with Schmidt decomposition as
\beq
\label{eq:pure-schdec}
\ket{\psi}=\sum_{i=1}^d \lambda_i \ket{ii},~~\lambda_i>0,~~\sum_{i=1}^d \lambda_i^2 =1.
\eeq
By computing, its PT is given by 
\beq
\label{eq:pure-pt-dec}
\ketbra{\psi}^\Gamma=\sum_{i,j=1}^d \lambda_i\lambda_j \ket{ij}\bra{ji}.
\eeq
One can verify that $\ket{ii}$ and $\ket{ij}\pm\ket{ji}$ are the eigenvectors of $\ketbra{\psi}^\Gamma$ with the corresponding eigenvalues
\beq
\label{eq:pure-pt-eig}
\bal
\lambda_i^2,\quad &\forall i=1,\cdots,d,\\
\pm \lambda_i\lambda_j,\quad &\forall 1\leq i<j\leq d.
\eal
\eeq
With the restriction $\sum_{i=1}^d\lambda_i^2=1$, one can verify that the following inequality holds,
\beq
-\frac{1}{2}\leq \lambda_{\min}(\ketbra{\psi}^\Gamma)\leq \lambda_{\max}(\ketbra{\psi}^\Gamma)\leq 1.
\eeq
Furthermore, $\lambda_{\min}$ reaches $-\frac{1}{2}$ if and only if $\lambda_1^2=\lambda_2^2=\frac{1}{2}$. In other words, $\lambda_{\min}(\ketbra{\psi}^\Gamma)=-\frac{1}{2}$ if and only if $\ket{\psi}$ is a two-qubit maximally entangled state. For the upper bound, one may check that $\lambda_{\max}(\ketbra{\psi}^\Gamma)=1$ if and only if $\ket{\psi}$ is a product state. Thus, no pure state can saturate both the bounds. 

We next consider a mixed state $\rho$ whose spectral decomposition is given by
\beq
\label{eq:mix-specdec}
\rho=\sum_k p_k\ketbra{\psi_k}.
\eeq
It follows that its PT $\rho^\Gamma$ is expressed as 
\beq
\label{eq:mix-pt-dec}
\rho^\Gamma=\sum_k p_k\ketbra{\psi_k}^\Gamma.
\eeq
Then we have 
\beq
\label{eq:lmin-1}
\bal
\lambda_{\min}(\rho^\Gamma)&\geq\sum_k\lambda_{\min}\left(p_k\ketbra{\psi_k}^\Gamma\right) \\
&\geq\sum_k p_k (-\frac{1}{2}) \\
&=-\frac{1}{2},
\eal
\eeq
where the first inequality follows from the lower bound in Eq. \eqref{eq:eigenineq-1} of Lemma \ref{le:eigenineq}. By observing Eq. \eqref{eq:lmin-1}, $\lambda_{\min}(\rho^\Gamma)=-\frac{1}{2}$ if and only if the two inequalities in \eqref{eq:lmin-1} both become equalities. The second equality occurs if and only if each $\ket{\psi_k}$ has Schmidt rank two and equal Schmidt coefficients. More specifically, by Lemma \ref{le:eigenineq}, we conclude that the first equality occurs if and only if there is a common vector $\ket{x}$ such that $\left(\ketbra{\psi_k}^\Gamma\right)\ket{x}=-\frac{1}{2}\ket{x}$ for each $k$. According to the analysis of the PT of a pure state, the common vector $\ket{x}$ has to be $\ket{01}-\ket{10}$, and each $\ket{\psi_k}$ has to be $\frac{1}{\sqrt 2}(\ket{00}+\ket{11})$. Thus, $\lambda_{\min}(\rho^\Gamma)=-\frac{1}{2}$ if and only if $\rho$ is a two-qubit maximally entangled state. Similarly, utilizing the dual inequality for $\lambda_{\max}$, we have 
\beq
\label{eq:lmax-1}
\bal
\lambda_{\max}(\rho^\Gamma)&\leq\sum_k\lambda_{\max}\left(p_k\ketbra{\psi_k}^\Gamma\right) \\
&\leq \sum_k p_k\cdot 1 \\
&=1,
\eal
\eeq
where the first inequality follows from the upper bound in Eq. \eqref{eq:eigenineq-1} of Lemma \ref{le:eigenineq}. By observing Eq. \eqref{eq:lmax-1}, $\lambda_{\max}(\rho^\Gamma)=1$ if and only if the two inequalities in \eqref{eq:lmax-1} both become equalities. The second equality occurs if and only if each $\ket{\psi_k}$ is a pure product state, namely $\ket{a_k,b_k}$. By Lemma \ref{le:eigenineq}, we conclude that the first equality occurs if and only if there is a common vector $\ket{x}$ such that $\left(\ketbra{\psi_k}^\Gamma\right)\ket{x}=\ket{x}$ for each $k$. It implies that each product state $\ket{\psi_k}$ has to be the same. Thus, $\lambda_{\max}(\rho^\Gamma)=1$ if and only if $\rho$ is a pure product state, namely $\ketbra{a,b}$. 

This completes the proof.
\qed

\section{Proofs of results in Sec. \ref{sec:rangelu}.}
\label{sec:proof2}

In this section, we show some results in Sec. \ref{sec:rangelu} in detail. First, we give the proof of Lemma \ref{le:crlu-1}. As preliminaries, we present the following lemma which allows the unitary freedom in the ensemble for density matrices.

\begin{lemma}
\label{le:u-freedom}
A density matrix can be decomposed as $\rho=\sum_i\ketbra{\tilde{\psi}_i}=\sum_i\ketbra{\tilde{\varphi}_i}$ if and only if 
\beq
\label{eq:u-freedom-main}
\ket{\tilde{\psi}_i}=\sum_j u_{ij}\ket{\tilde{\varphi}_j},
\eeq
where $u_{ij}$ is the $(i,j)$ element of a unitary matrix. 
\end{lemma}

With Lemma \ref{le:u-freedom} we are ready to prove Lemma \ref{le:crlu-1} as below.

\textbf{Proof of Lemma \ref{le:crlu-1}.}
For two bipartite states $\rho,\sigma$ with spectral decompositions as $\rho=\sum_{j=1}^r \lambda_j P_j$ and $\sigma=\sum_{j=1}^r \lambda_j Q_j$, it follows from Lemma \ref{le:spectral-dec} that $\rho\sim\sigma$ if and only if for each $j$, $P_j$ and $Q_j$ are SLU equivalent, i.e. LU equivalent by a common LU operator. We claim that for two given projectors $P,Q$, $P\sim Q$ if and only if their ranges $\cR(P),\cR(Q)$ are LU equivalent, i.e. $U\ox V$ such that $(U\ox V)\cR(P)=\cR(Q)$ for some LU operator $U\ox V$. The ``Only if'' part of this claim can be verified directly. We next show the ``If'' part of this claim. Let $\{\ket{a_1},\cdots,\ket{a_r}\}$ be an orthonormal basis of $\cR(P)$. According to the condition $(U\ox V)\cR(P)=\cR(Q)$, it follows that $(U\ox V)\ket{a_j}\in\cR(Q)$ for each $j$. Denote that $\ket{b_j}=(U\ox V)\ket{a_j}$ for each $j$, and thus $\cR(Q)=\lin\{\ket{b_1},\cdots,\ket{b_r}\}$. Applying Lemma \ref{le:u-freedom}, we know that $Q$ has the following decomposition:
\beq
\label{eq:wdec-1}
Q=\sum_{j=1}^r (\sum_{k=1}^r u_{jk}\ket{b_k})(\sum_{l=1}^r u_{jl}^*\bra{b_l}),
\eeq
where $u_{jk}$ is the $(j,k)$ element of a unitary matrix. It follows by simplifying Eq. \eqref{eq:wdec-1} that $Q$ has the spectral decomposition as $Q=\sum_{j=1}^r \ketbra{b_j}$. Due to $\ket{b_j}=(U\ox V)\ket{a_j}$ for each $j$, we conclude that $P\sim Q$. By the claim above, we conclude that $\rho\sim\sigma$ if and only if for each $j$, the ranges of $P_j,Q_j$ are SLU equivalent, i.e. $\cR(P_j)=(U\ox V)\cR(Q_j)$ by a common LU operator $U\ox V$ for each $j$. This completes the proof.
\qed

Second, we show Lemma \ref{le:4prjs} in detail as below.

\textbf{Proof of Lemma \ref{le:4prjs}.}
According to the precondition, we may assume that
\beq
\label{eq:pslu-1}
\bal
P_1&=(U_1\ox V_1)Q_1(U_1\ox V_1)^\dg,\\
P_2&=(U_1\ox V_1)Q_2(U_1\ox V_1)^\dg,
\eal
\eeq
\beq
\label{eq:pslu-2}
\bal
P_1&=(U_2\ox V_2)Q_1(U_2\ox V_2)^\dg,\\
P_3&=(U_2\ox V_2)Q_3(U_2\ox V_2)^\dg,
\eal
\eeq
\beq
\label{eq:pslu-3}
\bal
P_2&=(U_3\ox V_3)Q_2(U_3\ox V_3)^\dg,\\
P_3&=(U_3\ox V_3)Q_3(U_3\ox V_3)^\dg,
\eal
\eeq
with local unitary $U_j\ox V_j,~j=1,2,3$. By simultaneously applying any LU operation to $Q_i$'s, it follows that
\beq
\label{eq:pslu-4}
\bal
\tilde{Q}_1&=(U\ox V)Q_1(U\ox V)^\dg,\\
\tilde{Q}_2&=(U\ox V)Q_2(U\ox V)^\dg,\\
\tilde{Q}_3&=(U\ox V)Q_3(U\ox V)^\dg.
\eal
\eeq
Next, we show the two assertions as follows.

(i) By Eq. \eqref{eq:pslu-1} we have
\beq
\label{eq:pslu-1.1}
\bal
P_1&=(U_1U^\dg\ox V_1V^\dg)\tilde{Q}_1(U U_1^\dg\ox VV_1^\dg),\\
P_2&=(U_1U^\dg\ox V_1V^\dg)\tilde{Q}_2(U U_1^\dg\ox VV_1^\dg). 
\eal
\eeq
By Eq. \eqref{eq:pslu-1.1} 
we conclude that $(P_1,P_2)$ is SLU equivalent to $(\tilde{Q}_1,\tilde{Q}_2)$, namely $(P_1,P_2)\sim_s(\tilde{Q}_1,\tilde{Q}_2)$. Analagously, one can verify that $(P_1,P_3)\sim_s(\tilde{Q}_1,\tilde{Q}_3)$ by virtue of Eqs. \eqref{eq:pslu-2} and \eqref{eq:pslu-4}; and $(P_2,P_3)\sim_s(\tilde{Q}_2,\tilde{Q}_3)$ by virtue of Eqs. \eqref{eq:pslu-3} and \eqref{eq:pslu-4}. Then the assertion (i) holds.

(ii) We show the ``If'' part, and the ``Only if'' part can be verified similarly, due to the dual relation. Assume that $(P_1,P_2,P_3)\sim_s(\tilde{Q}_1,\tilde{Q}_2,\tilde{Q}_3)$. It means there exists a common LU operator $W\ox X$ such that $P_j=(W\ox X)\tilde{Q}_j(W\ox X)^\dg$ for each $j$. By Eq. \eqref{eq:pslu-4} we obtain
\beq
\label{eq:pslu-5}
\bal
P_1&=(WU\ox XV)Q_1(WU\ox XV)^\dg,\\
P_2&=(WU\ox XV)Q_2(WU\ox XV)^\dg,\\
P_3&=(WU\ox XV)Q_3(WU\ox XV)^\dg.
\eal
\eeq
It follows that $(P_1,P_2,P_3)\sim_s(Q_1,Q_2,Q_3)$. Then the assertion (ii) holds.

This completes the proof.
\qed


\bibliography{witness}

\begin{thebibliography}{34}%
\makeatletter
\providecommand \@ifxundefined [1]{%
 \@ifx{#1\undefined}
}%
\providecommand \@ifnum [1]{%
 \ifnum #1\expandafter \@firstoftwo
 \else \expandafter \@secondoftwo
 \fi
}%
\providecommand \@ifx [1]{%
 \ifx #1\expandafter \@firstoftwo
 \else \expandafter \@secondoftwo
 \fi
}%
\providecommand \natexlab [1]{#1}%
\providecommand \enquote  [1]{``#1''}%
\providecommand \bibnamefont  [1]{#1}%
\providecommand \bibfnamefont [1]{#1}%
\providecommand \citenamefont [1]{#1}%
\providecommand \href@noop [0]{\@secondoftwo}%
\providecommand \href [0]{\begingroup \@sanitize@url \@href}%
\providecommand \@href[1]{\@@startlink{#1}\@@href}%
\providecommand \@@href[1]{\endgroup#1\@@endlink}%
\providecommand \@sanitize@url [0]{\catcode `\\12\catcode `\$12\catcode `\&12\catcode `\#12\catcode `\^12\catcode `\_12\catcode `\%12\relax}%
\providecommand \@@startlink[1]{}%
\providecommand \@@endlink[0]{}%
\providecommand \url  [0]{\begingroup\@sanitize@url \@url }%
\providecommand \@url [1]{\endgroup\@href {#1}{\urlprefix }}%
\providecommand \urlprefix  [0]{URL }%
\providecommand \Eprint [0]{\href }%
\providecommand \doibase [0]{https://doi.org/}%
\providecommand \selectlanguage [0]{\@gobble}%
\providecommand \bibinfo  [0]{\@secondoftwo}%
\providecommand \bibfield  [0]{\@secondoftwo}%
\providecommand \translation [1]{[#1]}%
\providecommand \BibitemOpen [0]{}%
\providecommand \bibitemStop [0]{}%
\providecommand \bibitemNoStop [0]{.\EOS\space}%
\providecommand \EOS [0]{\spacefactor3000\relax}%
\providecommand \BibitemShut  [1]{\csname bibitem#1\endcsname}%
\let\auto@bib@innerbib\@empty
\bibitem [{\citenamefont {Kraus}(2010{\natexlab{a}})}]{pslu-2010}%
  \BibitemOpen
  \bibfield  {author} {\bibinfo {author} {\bibfnamefont {B.}~\bibnamefont {Kraus}},\ }\href {https://doi.org/10.1103/PhysRevA.82.032121} {\bibfield  {journal} {\bibinfo  {journal} {Phys. Rev. A}\ }\textbf {\bibinfo {volume} {82}},\ \bibinfo {pages} {032121} (\bibinfo {year} {2010}{\natexlab{a}})}\BibitemShut {NoStop}%
\bibitem [{\citenamefont {Dur}\ \emph {et~al.}(2000)\citenamefont {Dur}, \citenamefont {Vidal},\ and\ \citenamefont {Cirac}}]{3qubitinequiv2000}%
  \BibitemOpen
  \bibfield  {author} {\bibinfo {author} {\bibfnamefont {W.}~\bibnamefont {Dur}}, \bibinfo {author} {\bibfnamefont {G.}~\bibnamefont {Vidal}},\ and\ \bibinfo {author} {\bibfnamefont {J.~I.}\ \bibnamefont {Cirac}},\ }\href {https://doi.org/10.1103/PhysRevA.62.062314} {\bibfield  {journal} {\bibinfo  {journal} {Phys. Rev. A}\ }\textbf {\bibinfo {volume} {62}},\ \bibinfo {pages} {062314} (\bibinfo {year} {2000})}\BibitemShut {NoStop}%
\bibitem [{\citenamefont {Chen}\ and\ \citenamefont {Hayashi}(2011)}]{MCSLOCC2011}%
  \BibitemOpen
  \bibfield  {author} {\bibinfo {author} {\bibfnamefont {L.}~\bibnamefont {Chen}}\ and\ \bibinfo {author} {\bibfnamefont {M.}~\bibnamefont {Hayashi}},\ }\href {https://doi.org/10.1103/PhysRevA.83.022331} {\bibfield  {journal} {\bibinfo  {journal} {Phys. Rev. A}\ }\textbf {\bibinfo {volume} {83}},\ \bibinfo {pages} {022331} (\bibinfo {year} {2011})}\BibitemShut {NoStop}%
\bibitem [{\citenamefont {Brand\~ao}\ and\ \citenamefont {Gour}(2015)}]{rvqrt2015}%
  \BibitemOpen
  \bibfield  {author} {\bibinfo {author} {\bibfnamefont {F.~G. S.~L.}\ \bibnamefont {Brand\~ao}}\ and\ \bibinfo {author} {\bibfnamefont {G.}~\bibnamefont {Gour}},\ }\href {https://doi.org/10.1103/PhysRevLett.115.070503} {\bibfield  {journal} {\bibinfo  {journal} {Phys. Rev. Lett.}\ }\textbf {\bibinfo {volume} {115}},\ \bibinfo {pages} {070503} (\bibinfo {year} {2015})}\BibitemShut {NoStop}%
\bibitem [{\citenamefont {Marvian}(2022)}]{luqc-2022}%
  \BibitemOpen
  \bibfield  {author} {\bibinfo {author} {\bibfnamefont {I.}~\bibnamefont {Marvian}},\ }\href {https://doi.org/10.1038/s41567-021-01464-0} {\bibfield  {journal} {\bibinfo  {journal} {Nature Physics}\ }\textbf {\bibinfo {volume} {18}},\ \bibinfo {pages} {283} (\bibinfo {year} {2022})}\BibitemShut {NoStop}%
\bibitem [{\citenamefont {Piroli}\ \emph {et~al.}(2021)\citenamefont {Piroli}, \citenamefont {Styliaris},\ and\ \citenamefont {Cirac}}]{qcLOCC2021}%
  \BibitemOpen
  \bibfield  {author} {\bibinfo {author} {\bibfnamefont {L.}~\bibnamefont {Piroli}}, \bibinfo {author} {\bibfnamefont {G.}~\bibnamefont {Styliaris}},\ and\ \bibinfo {author} {\bibfnamefont {J.~I.}\ \bibnamefont {Cirac}},\ }\href {https://doi.org/10.1103/PhysRevLett.127.220503} {\bibfield  {journal} {\bibinfo  {journal} {Phys. Rev. Lett.}\ }\textbf {\bibinfo {volume} {127}},\ \bibinfo {pages} {220503} (\bibinfo {year} {2021})}\BibitemShut {NoStop}%
\bibitem [{\citenamefont {Oszmaniec}\ \emph {et~al.}(2016)\citenamefont {Oszmaniec}, \citenamefont {Augusiak}, \citenamefont {Gogolin}, \citenamefont {Ko\l{}ody\ifmmode~\acute{n}\else \'{n}\fi{}ski}, \citenamefont {Ac\'{\i}n},\ and\ \citenamefont {Lewenstein}}]{metrology2016}%
  \BibitemOpen
  \bibfield  {author} {\bibinfo {author} {\bibfnamefont {M.}~\bibnamefont {Oszmaniec}}, \bibinfo {author} {\bibfnamefont {R.}~\bibnamefont {Augusiak}}, \bibinfo {author} {\bibfnamefont {C.}~\bibnamefont {Gogolin}}, \bibinfo {author} {\bibfnamefont {J.}~\bibnamefont {Ko\l{}ody\ifmmode~\acute{n}\else \'{n}\fi{}ski}}, \bibinfo {author} {\bibfnamefont {A.}~\bibnamefont {Ac\'{\i}n}},\ and\ \bibinfo {author} {\bibfnamefont {M.}~\bibnamefont {Lewenstein}},\ }\href {https://doi.org/10.1103/PhysRevX.6.041044} {\bibfield  {journal} {\bibinfo  {journal} {Phys. Rev. X}\ }\textbf {\bibinfo {volume} {6}},\ \bibinfo {pages} {041044} (\bibinfo {year} {2016})}\BibitemShut {NoStop}%
\bibitem [{\citenamefont {Acin}\ \emph {et~al.}(2000)\citenamefont {Acin}, \citenamefont {Andrianov}, \citenamefont {Costa}, \citenamefont {Jane}, \citenamefont {Latorre},\ and\ \citenamefont {Tarrach}}]{3qubitlu2000}%
  \BibitemOpen
  \bibfield  {author} {\bibinfo {author} {\bibfnamefont {A.}~\bibnamefont {Acin}}, \bibinfo {author} {\bibfnamefont {A.}~\bibnamefont {Andrianov}}, \bibinfo {author} {\bibfnamefont {L.}~\bibnamefont {Costa}}, \bibinfo {author} {\bibfnamefont {E.}~\bibnamefont {Jane}}, \bibinfo {author} {\bibfnamefont {J.~I.}\ \bibnamefont {Latorre}},\ and\ \bibinfo {author} {\bibfnamefont {R.}~\bibnamefont {Tarrach}},\ }\href {https://doi.org/10.1103/PhysRevLett.85.1560} {\bibfield  {journal} {\bibinfo  {journal} {Phys. Rev. Lett.}\ }\textbf {\bibinfo {volume} {85}},\ \bibinfo {pages} {1560} (\bibinfo {year} {2000})}\BibitemShut {NoStop}%
\bibitem [{\citenamefont {Kraus}(2010{\natexlab{b}})}]{luequiv2010}%
  \BibitemOpen
  \bibfield  {author} {\bibinfo {author} {\bibfnamefont {B.}~\bibnamefont {Kraus}},\ }\href {https://doi.org/10.1103/PhysRevLett.104.020504} {\bibfield  {journal} {\bibinfo  {journal} {Phys. Rev. Lett.}\ }\textbf {\bibinfo {volume} {104}},\ \bibinfo {pages} {020504} (\bibinfo {year} {2010}{\natexlab{b}})}\BibitemShut {NoStop}%
\bibitem [{\citenamefont {Liu}\ \emph {et~al.}(2012)\citenamefont {Liu}, \citenamefont {Li}, \citenamefont {Li},\ and\ \citenamefont {Qiao}}]{mpsinequivlu2012}%
  \BibitemOpen
  \bibfield  {author} {\bibinfo {author} {\bibfnamefont {B.}~\bibnamefont {Liu}}, \bibinfo {author} {\bibfnamefont {J.-L.}\ \bibnamefont {Li}}, \bibinfo {author} {\bibfnamefont {X.}~\bibnamefont {Li}},\ and\ \bibinfo {author} {\bibfnamefont {C.-F.}\ \bibnamefont {Qiao}},\ }\href {https://doi.org/10.1103/PhysRevLett.108.050501} {\bibfield  {journal} {\bibinfo  {journal} {Phys. Rev. Lett.}\ }\textbf {\bibinfo {volume} {108}},\ \bibinfo {pages} {050501} (\bibinfo {year} {2012})}\BibitemShut {NoStop}%
\bibitem [{\citenamefont {Makhlin}(2002)}]{Makhlin2002}%
  \BibitemOpen
  \bibfield  {author} {\bibinfo {author} {\bibfnamefont {Y.}~\bibnamefont {Makhlin}},\ }\href {https://doi.org/10.1023/A:1022144002391} {\bibfield  {journal} {\bibinfo  {journal} {Quantum Information Processing}\ }\textbf {\bibinfo {volume} {1}},\ \bibinfo {pages} {243} (\bibinfo {year} {2002})}\BibitemShut {NoStop}%
\bibitem [{\citenamefont {Zhang}\ \emph {et~al.}(2013)\citenamefont {Zhang}, \citenamefont {Zhao}, \citenamefont {Li}, \citenamefont {Fei},\ and\ \citenamefont {Li-Jost}}]{msinv2013}%
  \BibitemOpen
  \bibfield  {author} {\bibinfo {author} {\bibfnamefont {T.-G.}\ \bibnamefont {Zhang}}, \bibinfo {author} {\bibfnamefont {M.-J.}\ \bibnamefont {Zhao}}, \bibinfo {author} {\bibfnamefont {M.}~\bibnamefont {Li}}, \bibinfo {author} {\bibfnamefont {S.-M.}\ \bibnamefont {Fei}},\ and\ \bibinfo {author} {\bibfnamefont {X.}~\bibnamefont {Li-Jost}},\ }\href {https://doi.org/10.1103/PhysRevA.88.042304} {\bibfield  {journal} {\bibinfo  {journal} {Phys. Rev. A}\ }\textbf {\bibinfo {volume} {88}},\ \bibinfo {pages} {042304} (\bibinfo {year} {2013})}\BibitemShut {NoStop}%
\bibitem [{\citenamefont {Li}\ \emph {et~al.}(2014)\citenamefont {Li}, \citenamefont {Zhang}, \citenamefont {Fei}, \citenamefont {Li-Jost},\ and\ \citenamefont {Jing}}]{mqbitinv2014}%
  \BibitemOpen
  \bibfield  {author} {\bibinfo {author} {\bibfnamefont {M.}~\bibnamefont {Li}}, \bibinfo {author} {\bibfnamefont {T.}~\bibnamefont {Zhang}}, \bibinfo {author} {\bibfnamefont {S.-M.}\ \bibnamefont {Fei}}, \bibinfo {author} {\bibfnamefont {X.}~\bibnamefont {Li-Jost}},\ and\ \bibinfo {author} {\bibfnamefont {N.}~\bibnamefont {Jing}},\ }\href {https://doi.org/10.1103/PhysRevA.89.062325} {\bibfield  {journal} {\bibinfo  {journal} {Phys. Rev. A}\ }\textbf {\bibinfo {volume} {89}},\ \bibinfo {pages} {062325} (\bibinfo {year} {2014})}\BibitemShut {NoStop}%
\bibitem [{\citenamefont {Martins}(2015)}]{mqbitiff2015}%
  \BibitemOpen
  \bibfield  {author} {\bibinfo {author} {\bibfnamefont {A.~M.}\ \bibnamefont {Martins}},\ }\href {https://doi.org/10.1103/PhysRevA.91.042308} {\bibfield  {journal} {\bibinfo  {journal} {Phys. Rev. A}\ }\textbf {\bibinfo {volume} {91}},\ \bibinfo {pages} {042308} (\bibinfo {year} {2015})}\BibitemShut {NoStop}%
\bibitem [{\citenamefont {Jing}\ \emph {et~al.}(2015)\citenamefont {Jing}, \citenamefont {Fei}, \citenamefont {Li}, \citenamefont {Li-Jost},\ and\ \citenamefont {Zhang}}]{mqbitinv2015}%
  \BibitemOpen
  \bibfield  {author} {\bibinfo {author} {\bibfnamefont {N.}~\bibnamefont {Jing}}, \bibinfo {author} {\bibfnamefont {S.-M.}\ \bibnamefont {Fei}}, \bibinfo {author} {\bibfnamefont {M.}~\bibnamefont {Li}}, \bibinfo {author} {\bibfnamefont {X.}~\bibnamefont {Li-Jost}},\ and\ \bibinfo {author} {\bibfnamefont {T.}~\bibnamefont {Zhang}},\ }\href {https://doi.org/10.1103/PhysRevA.92.022306} {\bibfield  {journal} {\bibinfo  {journal} {Phys. Rev. A}\ }\textbf {\bibinfo {volume} {92}},\ \bibinfo {pages} {022306} (\bibinfo {year} {2015})}\BibitemShut {NoStop}%
\bibitem [{\citenamefont {Sun}\ \emph {et~al.}(2017)\citenamefont {Sun}, \citenamefont {Fei},\ and\ \citenamefont {Wang}}]{3qbitinv2017}%
  \BibitemOpen
  \bibfield  {author} {\bibinfo {author} {\bibfnamefont {B.-Z.}\ \bibnamefont {Sun}}, \bibinfo {author} {\bibfnamefont {S.-M.}\ \bibnamefont {Fei}},\ and\ \bibinfo {author} {\bibfnamefont {Z.-X.}\ \bibnamefont {Wang}},\ }\href {https://doi.org/10.1038/s41598-017-04717-2} {\bibfield  {journal} {\bibinfo  {journal} {Scientific Reports}\ }\textbf {\bibinfo {volume} {7}},\ \bibinfo {pages} {4869} (\bibinfo {year} {2017})}\BibitemShut {NoStop}%
\bibitem [{\citenamefont {Zhou}\ \emph {et~al.}(2024)\citenamefont {Zhou}, \citenamefont {Zhen}, \citenamefont {Xu}, \citenamefont {Zhao}, \citenamefont {Yang}, \citenamefont {Fei}, \citenamefont {Li}, \citenamefont {Liu},\ and\ \citenamefont {Chen}}]{msinv2024}%
  \BibitemOpen
  \bibfield  {author} {\bibinfo {author} {\bibfnamefont {Q.}~\bibnamefont {Zhou}}, \bibinfo {author} {\bibfnamefont {Y.-Z.}\ \bibnamefont {Zhen}}, \bibinfo {author} {\bibfnamefont {X.-Y.}\ \bibnamefont {Xu}}, \bibinfo {author} {\bibfnamefont {S.}~\bibnamefont {Zhao}}, \bibinfo {author} {\bibfnamefont {W.-L.}\ \bibnamefont {Yang}}, \bibinfo {author} {\bibfnamefont {S.-M.}\ \bibnamefont {Fei}}, \bibinfo {author} {\bibfnamefont {L.}~\bibnamefont {Li}}, \bibinfo {author} {\bibfnamefont {N.-L.}\ \bibnamefont {Liu}},\ and\ \bibinfo {author} {\bibfnamefont {K.}~\bibnamefont {Chen}},\ }\href {https://doi.org/10.1103/PhysRevA.109.022427} {\bibfield  {journal} {\bibinfo  {journal} {Phys. Rev. A}\ }\textbf {\bibinfo {volume} {109}},\ \bibinfo {pages} {022427} (\bibinfo {year} {2024})}\BibitemShut {NoStop}%
\bibitem [{\citenamefont {Zhang}\ \emph {et~al.}(2024)\citenamefont {Zhang}, \citenamefont {Xie},\ and\ \citenamefont {Tao}}]{lininv2024}%
  \BibitemOpen
  \bibfield  {author} {\bibinfo {author} {\bibfnamefont {L.}~\bibnamefont {Zhang}}, \bibinfo {author} {\bibfnamefont {B.}~\bibnamefont {Xie}},\ and\ \bibinfo {author} {\bibfnamefont {Y.}~\bibnamefont {Tao}},\ }\href {https://arxiv.org/abs/2412.17237} {\bibinfo {title} {Local unitarty equivalence and entanglement by bargmann invariants}} (\bibinfo {year} {2024}),\ \Eprint {https://arxiv.org/abs/2412.17237} {arXiv:2412.17237 [quant-ph]} \BibitemShut {NoStop}%
\bibitem [{\citenamefont {Chen}\ and\ \citenamefont {Yu}(2015)}]{decugate15}%
  \BibitemOpen
  \bibfield  {author} {\bibinfo {author} {\bibfnamefont {L.}~\bibnamefont {Chen}}\ and\ \bibinfo {author} {\bibfnamefont {L.}~\bibnamefont {Yu}},\ }\href {https://doi.org/10.1103/PhysRevA.91.032308} {\bibfield  {journal} {\bibinfo  {journal} {Phys. Rev. A}\ }\textbf {\bibinfo {volume} {91}},\ \bibinfo {pages} {032308} (\bibinfo {year} {2015})}\BibitemShut {NoStop}%
\bibitem [{\citenamefont {Shen}\ \emph {et~al.}(2022)\citenamefont {Shen}, \citenamefont {Chen},\ and\ \citenamefont {Yu}}]{ugatesyi22}%
  \BibitemOpen
  \bibfield  {author} {\bibinfo {author} {\bibfnamefont {Y.}~\bibnamefont {Shen}}, \bibinfo {author} {\bibfnamefont {L.}~\bibnamefont {Chen}},\ and\ \bibinfo {author} {\bibfnamefont {L.}~\bibnamefont {Yu}},\ }\href {https://doi.org/10.1088/1751-8121/aca36b} {\bibfield  {journal} {\bibinfo  {journal} {Journal of Physics A: Mathematical and Theoretical}\ }\textbf {\bibinfo {volume} {55}},\ \bibinfo {pages} {465302} (\bibinfo {year} {2022})}\BibitemShut {NoStop}%
\bibitem [{\citenamefont {Chruściński}\ and\ \citenamefont {Sarbicki}(2014)}]{ewreview2014}%
  \BibitemOpen
  \bibfield  {author} {\bibinfo {author} {\bibfnamefont {D.}~\bibnamefont {Chruściński}}\ and\ \bibinfo {author} {\bibfnamefont {G.}~\bibnamefont {Sarbicki}},\ }\href {https://doi.org/10.1088/1751-8113/47/48/483001} {\bibfield  {journal} {\bibinfo  {journal} {Journal of Physics A: Mathematical and Theoretical}\ }\textbf {\bibinfo {volume} {47}},\ \bibinfo {pages} {483001} (\bibinfo {year} {2014})}\BibitemShut {NoStop}%
\bibitem [{\citenamefont {Shen}\ \emph {et~al.}(2024)\citenamefont {Shen}, \citenamefont {Chen},\ and\ \citenamefont {Bian}}]{rew2024}%
  \BibitemOpen
  \bibfield  {author} {\bibinfo {author} {\bibfnamefont {Y.}~\bibnamefont {Shen}}, \bibinfo {author} {\bibfnamefont {L.}~\bibnamefont {Chen}},\ and\ \bibinfo {author} {\bibfnamefont {Z.}~\bibnamefont {Bian}},\ }\href {https://arxiv.org/abs/2408.08574} {\bibinfo {title} {The detection power of real entanglement witnesses under local unitary equivalence}} (\bibinfo {year} {2024}),\ \Eprint {https://arxiv.org/abs/2408.08574} {arXiv:2408.08574 [quant-ph]} \BibitemShut {NoStop}%
\bibitem [{\citenamefont {Chen}\ and\ \citenamefont {Yu}(2016)}]{ep2016-2}%
  \BibitemOpen
  \bibfield  {author} {\bibinfo {author} {\bibfnamefont {L.}~\bibnamefont {Chen}}\ and\ \bibinfo {author} {\bibfnamefont {L.}~\bibnamefont {Yu}},\ }\href {https://doi.org/10.1103/PhysRevA.94.022307} {\bibfield  {journal} {\bibinfo  {journal} {Phys. Rev. A}\ }\textbf {\bibinfo {volume} {94}},\ \bibinfo {pages} {022307} (\bibinfo {year} {2016})}\BibitemShut {NoStop}%
\bibitem [{\citenamefont {Qiu}\ \emph {et~al.}(2025)\citenamefont {Qiu}, \citenamefont {Song},\ and\ \citenamefont {Chen}}]{mep2025}%
  \BibitemOpen
  \bibfield  {author} {\bibinfo {author} {\bibfnamefont {X.}~\bibnamefont {Qiu}}, \bibinfo {author} {\bibfnamefont {Z.}~\bibnamefont {Song}},\ and\ \bibinfo {author} {\bibfnamefont {L.}~\bibnamefont {Chen}},\ }\href {https://doi.org/10.1103/PhysRevA.111.022407} {\bibfield  {journal} {\bibinfo  {journal} {Phys. Rev. A}\ }\textbf {\bibinfo {volume} {111}},\ \bibinfo {pages} {022407} (\bibinfo {year} {2025})}\BibitemShut {NoStop}%
\bibitem [{\citenamefont {Liang}\ \emph {et~al.}(2024)\citenamefont {Liang}, \citenamefont {Yan}, \citenamefont {Si},\ and\ \citenamefont {Chen}}]{inertia24}%
  \BibitemOpen
  \bibfield  {author} {\bibinfo {author} {\bibfnamefont {Y.}~\bibnamefont {Liang}}, \bibinfo {author} {\bibfnamefont {J.}~\bibnamefont {Yan}}, \bibinfo {author} {\bibfnamefont {D.}~\bibnamefont {Si}},\ and\ \bibinfo {author} {\bibfnamefont {L.}~\bibnamefont {Chen}},\ }\href {https://doi.org/10.1088/1751-8121/ad3056} {\bibfield  {journal} {\bibinfo  {journal} {Journal of Physics A: Mathematical and Theoretical}\ }\textbf {\bibinfo {volume} {57}},\ \bibinfo {pages} {125203} (\bibinfo {year} {2024})}\BibitemShut {NoStop}%
\bibitem [{\citenamefont {Yu}\ \emph {et~al.}(2010)\citenamefont {Yu}, \citenamefont {Griffiths},\ and\ \citenamefont {Cohen}}]{eff-nug-2015}%
  \BibitemOpen
  \bibfield  {author} {\bibinfo {author} {\bibfnamefont {L.}~\bibnamefont {Yu}}, \bibinfo {author} {\bibfnamefont {R.~B.}\ \bibnamefont {Griffiths}},\ and\ \bibinfo {author} {\bibfnamefont {S.~M.}\ \bibnamefont {Cohen}},\ }\href {https://doi.org/10.1103/PhysRevA.81.062315} {\bibfield  {journal} {\bibinfo  {journal} {Phys. Rev. A}\ }\textbf {\bibinfo {volume} {81}},\ \bibinfo {pages} {062315} (\bibinfo {year} {2010})}\BibitemShut {NoStop}%
\bibitem [{\citenamefont {Amaro}\ and\ \citenamefont {M\"uller}(2020)}]{lewo2020}%
  \BibitemOpen
  \bibfield  {author} {\bibinfo {author} {\bibfnamefont {D.}~\bibnamefont {Amaro}}\ and\ \bibinfo {author} {\bibfnamefont {M.}~\bibnamefont {M\"uller}},\ }\href {https://doi.org/10.1103/PhysRevA.101.012317} {\bibfield  {journal} {\bibinfo  {journal} {Phys. Rev. A}\ }\textbf {\bibinfo {volume} {101}},\ \bibinfo {pages} {012317} (\bibinfo {year} {2020})}\BibitemShut {NoStop}%
\bibitem [{\citenamefont {Horodecki}(1997)}]{ppte1997}%
  \BibitemOpen
  \bibfield  {author} {\bibinfo {author} {\bibfnamefont {P.}~\bibnamefont {Horodecki}},\ }\href {https://doi.org/https://doi.org/10.1016/S0375-9601(97)00416-7} {\bibfield  {journal} {\bibinfo  {journal} {Physics Letters A}\ }\textbf {\bibinfo {volume} {232}},\ \bibinfo {pages} {333 } (\bibinfo {year} {1997})}\BibitemShut {NoStop}%
\bibitem [{\citenamefont {Folland}(1999)}]{ranalysis}%
  \BibitemOpen
  \bibfield  {author} {\bibinfo {author} {\bibfnamefont {G.~B.}\ \bibnamefont {Folland}},\ }\href@noop {} {\emph {\bibinfo {title} {Real Analysis: Modern Techniques and Their Applications}}},\ \bibinfo {edition} {2nd}\ ed.\ (\bibinfo  {publisher} {Wiley},\ \bibinfo {address} {USA},\ \bibinfo {year} {1999})\BibitemShut {NoStop}%
\bibitem [{\citenamefont {Rana}(2013)}]{Neigenv2013}%
  \BibitemOpen
  \bibfield  {author} {\bibinfo {author} {\bibfnamefont {S.}~\bibnamefont {Rana}},\ }\href {https://doi.org/10.1103/PhysRevA.87.054301} {\bibfield  {journal} {\bibinfo  {journal} {Phys. Rev. A}\ }\textbf {\bibinfo {volume} {87}},\ \bibinfo {pages} {054301} (\bibinfo {year} {2013})}\BibitemShut {NoStop}%
\bibitem [{\citenamefont {Bartlett}\ \emph {et~al.}(2007)\citenamefont {Bartlett}, \citenamefont {Rudolph},\ and\ \citenamefont {Spekkens}}]{ssr2007}%
  \BibitemOpen
  \bibfield  {author} {\bibinfo {author} {\bibfnamefont {S.~D.}\ \bibnamefont {Bartlett}}, \bibinfo {author} {\bibfnamefont {T.}~\bibnamefont {Rudolph}},\ and\ \bibinfo {author} {\bibfnamefont {R.~W.}\ \bibnamefont {Spekkens}},\ }\href {https://doi.org/10.1103/RevModPhys.79.555} {\bibfield  {journal} {\bibinfo  {journal} {Rev. Mod. Phys.}\ }\textbf {\bibinfo {volume} {79}},\ \bibinfo {pages} {555} (\bibinfo {year} {2007})}\BibitemShut {NoStop}%
\bibitem [{\citenamefont {Gour}\ and\ \citenamefont {Spekkens}(2008)}]{rfgour2008}%
  \BibitemOpen
  \bibfield  {author} {\bibinfo {author} {\bibfnamefont {G.}~\bibnamefont {Gour}}\ and\ \bibinfo {author} {\bibfnamefont {R.~W.}\ \bibnamefont {Spekkens}},\ }\href {https://doi.org/10.1088/1367-2630/10/3/033023} {\bibfield  {journal} {\bibinfo  {journal} {New J. Phys.}\ }\textbf {\bibinfo {volume} {10}},\ \bibinfo {pages} {033023} (\bibinfo {year} {2008})}\BibitemShut {NoStop}%
\bibitem [{\citenamefont {Gour}\ \emph {et~al.}(2009)\citenamefont {Gour}, \citenamefont {Marvian},\ and\ \citenamefont {Spekkens}}]{gour2009}%
  \BibitemOpen
  \bibfield  {author} {\bibinfo {author} {\bibfnamefont {G.}~\bibnamefont {Gour}}, \bibinfo {author} {\bibfnamefont {I.}~\bibnamefont {Marvian}},\ and\ \bibinfo {author} {\bibfnamefont {R.~W.}\ \bibnamefont {Spekkens}},\ }\href {https://doi.org/10.1103/PhysRevA.80.012307} {\bibfield  {journal} {\bibinfo  {journal} {Phys. Rev. A}\ }\textbf {\bibinfo {volume} {80}},\ \bibinfo {pages} {012307} (\bibinfo {year} {2009})}\BibitemShut {NoStop}%
\bibitem [{\citenamefont {Horn}\ and\ \citenamefont {Johnson}(2012)}]{bookmatrix}%
  \BibitemOpen
  \bibfield  {author} {\bibinfo {author} {\bibfnamefont {R.~A.}\ \bibnamefont {Horn}}\ and\ \bibinfo {author} {\bibfnamefont {C.~R.}\ \bibnamefont {Johnson}},\ }\href@noop {} {\emph {\bibinfo {title} {Matrix Analysis}}},\ \bibinfo {edition} {2nd}\ ed.\ (\bibinfo  {publisher} {Cambridge University Press},\ \bibinfo {address} {USA},\ \bibinfo {year} {2012})\BibitemShut {NoStop}%
\end{thebibliography}%

\end{document}